\documentclass{amsart}

\pdfoutput=1 

\usepackage[english]{babel}

\usepackage{arxiv}
\usepackage{array}
\usepackage{tabularx}

\usepackage{booktabs}
\usepackage{multirow}

\usepackage{graphicx}
\usepackage{geometry}
\usepackage{amsmath,amsthm,amssymb}
\usepackage{cancel}

\usepackage{thmtools}
\usepackage{thm-restate}
\usepackage{xcolor}

\theoremstyle{definition}
\theoremstyle{remark}

\theoremstyle{plain}
\newtheorem{theorem}{Theorem}

\newtheorem{lemma}[theorem]{Lemma}

\theoremstyle{definition}

\newtheorem{remark}[theorem]{Remark}
\newtheorem{definition}[theorem]{Definition}
\newtheorem{corollary}[theorem]{Corollary}

\usepackage{algorithm2e}

\title[Component twin-width as a parameter for BINARY-CSP]{Component twin-width as a parameter for BINARY-CSP and its semiring generalisations}
\author{Ambroise Baril}
\address[A. Baril]
   {Universit\'e de Lorraine, CNRS, LORIA \\
    F-54000 Nancy \\
    France
    }
\author{Miguel Couceiro}
\address[M. Couceiro]
   {Universit\'e de Lorraine, CNRS, LORIA \\
    F-54000 Nancy \\
    France
    }

\author{Victor Lagerkvist}
\address[V. Lagerkvist]
   {Dep. Computer and Information
      Science, \\Link\"opings  Universitet, Sweden}

\newcommand{\R}{{\mathcal{R}}}

\newcommand{\N}{\mathbb{N}}
\newcommand{\F}{\mathcal{F}}
\renewcommand{\P}{\mathcal{P}}
\renewcommand{\S}{\mathbb{S}}
\newcommand{\T}{\mathbb{T}}
\newcommand{\Rbar}{\overline{\mathbb{R}}}
\newcommand{\Nbar}{\overline{\mathbb{N}}}
\newcommand{\red}{\textcolor{red}{\textbf{e}}}

\usepackage{ifthen}

\begin{document}

\maketitle

\thispagestyle{empty}

\begin{abstract}
We investigate the fine-grained and the parameterized complexity of several generalizations of binary constraint satisfaction problems (BINARY-CSPs), that subsume variants of graph colouring problems. Our starting point is the observation that several algorithmic approaches that resulted in complexity upper bounds for these problems, share a common structure. 

We thus explore an algebraic approach relying on semirings that unifies different generalizations of BINARY-CSPs (such as the counting, the list, and the weighted versions), and that facilitates a general algorithmic approach to efficiently solving them. The latter is inspired by the (component) twin-width parameter introduced by Bonnet et al., which we generalize via edge-labelled graphs in order to formulate it to arbitrary binary constraints. We consider input instances with bounded component twin-width, as well as constraint templates of bounded component twin-width, and obtain an FPT algorithm as well as an improved, exponential-time algorithm, for broad classes of binary constraints.

We illustrate the advantages of this framework by instantiating our general algorithmic approach on several classes of problems (e.g., the $H$-coloring problem and its variants), and showing that it improves the best complexity upper bounds in the literature for several well-known problems.

\end{abstract}

\section{Introduction}

{\em Constraint satisfaction problems} (CSPs) are rooted in artificial intelligence and operations research and  provide a rich framework for encoding many different types of problems (for a wealth of applications, see e.g.\ the book by Rossi et al.~\cite{rossi2006handbook}).
A CSP example  of paramount importance is the {\em Boolean satisfiability problem}, SAT, which can be viewed as a constraint satisfaction problem over the Boolean domain. 
Another noteworthy CSP example is the $H$-COLORING problem~\cite{hell1990complexity} that asks whether there exists an homomorphism from an (undirected) input graph $G$ to the target graph $H$. Indeed, this problem can be formalized as a particular case of a CSP over binary constraints (BINARY-CSP): given a set of variables $V$ and a set of constraints of the form $R(u_1,u_2)$ for $(u_1,u_2) \in V^2$ and binary relations $R\subseteq D^2$, the objective is to map variables in $V$ to values in the domain $D$ such that every such constraint is satisfied. Such a mapping is then referred as a {\em solution} of the instance. When the relations belong to a fixed language $\Gamma$, this problem is usually denoted by BINARY-CSP$(\Gamma)$, and the decision problem is stated as the problem of deciding whether a solution exists. Clearly, the $H$-COLORING problems subsume the well known $q$-COLORING problem
with $H$ being the $q$-clique, which itself can be seen as the satisfaction of constraints expressed through the relation $\text{NEQ}_q=\{ (a,b)\in [q]^2 \mid a\neq b \}$.  Hell \& Ne{\v{s}}et{\v{r}}il~\cite{hell1990complexity} showed that the complexity of $H$-COLORING problems for a given graph $H$ is NP-complete whenever $H$ is not bipartite; otherwise, it is in P. This result can be seen as a particular case of the {\em CSP dichotomy theorem}
which states that every CSP problem (and, in particular, BINARY-CSP) is either in P or NP-complete~\cite{DBLP:journals/corr/abs-2007-09099,DBLP:journals/jacm/Zhuk20}.

For certain applications however, the basic decidability setting provides insufficient modelling power and one may be instead interested  in counting the number of solutions (we prefix the problem by \# to denote the corresponding counting problem). Also, counting problems may  be subjected to additional constraints such as lists, where a function $l:V\mapsto\P(D)$ is given and any solution $f$ has to satisfy $\forall u\in V,f(u)\in l(u)$, and costs, where sending $u\in V$ to $v\in D$ has a cost $C(u,v)$, with the goal of minimizing ${\scriptstyle \sum\limits_{u\in V}} C(u,f(u))$. This framework makes it possible to encode phase transition systems modelled  by partition functions, modeling problems such as counting $q$-particle Widom–Rowlinson configurations and counting Beach models, or the classical Ising model (for many more examples, see e.g.\  Dyer \& Greenhill~\cite{dyer2000complexity}). The complexity of these generalized problems has been the subject of intense research, starting with the complexity dichotomy for the \#$H$-COLORING problems by Dyer \& Greenhill~\cite{dyer2000complexity}, later extended to \#CSP by Bulatov~\cite{DBLP:journals/corr/abs-2007-09099}, and culminating into the dichotomy theorem by Cai \& Chen~\cite{cai2017} in the presence of weights.

In contrast to the classical CSP, which admits a rich tractability landscape~\cite{DBLP:journals/corr/abs-2007-09099,DBLP:journals/jacm/Zhuk20}, the associated counting problems are generally hard except for trivial cases. For example, \#$H$-COLORING is \#P-hard if $H$ is not a disjoint union of looped cliques and loopless complete bipartite graphs, and is in P otherwise~\cite{dyer2000complexity}. Thus, non-trivial templates $H$ of actual interest tend to result in NP-hard (\#)$H$-COLORING problems, necessitating tools and techniques from {\em parameterized} and {\em fine-grained} complexity.

Here, there is no lack of results for {\em specific} templates $H$. One of the most well-known results is likely Bj\"orklund et al.'s \cite{bjorklund2009set} dynamic programming algorithm which solves $q$-COLORING, as well as the  counting, summation, and optimization versions, in $O^*(2^n)$\footnote{The $O^*$ notation means that we ignore polynomial factors.} time. Another  result for $\#H$-COLORING from D{\'\i}az et al.~\cite{diaz2002counting} is the existence of an improved algorithm assuming that the input graph $G$ has tree-width $k$ and is given with a nice tree decomposition: for any graph $H$, $\#H$-COLORING can then be solved on $G$ in $O(|V_H|^{k+1}\min(k,n)\cdot n)$ time using only $O(|V_H|^{k+1}\log(n))$ additional space. 
The latter examples indicate that  bounding the class of input instances by parameters such as tree-width and the related  clique-width, can be applied more generally and with powerful algorithmic results.

More generally, two noteworthy extensions of $H$-COLORING can be defined  by restricting either the input or the target graphs with respect to a given parameter $k$. In the former, the goal is to obtain FPT algorithms, while one in the latter simply wishes to improve upon exhaustive search ($O^*(|D|^n)$), with a single-exponential running time $O^*(c^n)$ for a fixed $c$ depending only on $k$ representing the most desired outcome. Hence, restricting the class of input graphs yields a problem of interest in parameterized complexity while restricting the target graphs yields a problem with closer ties to fast exponential-time algorithms and fine-grained complexity. For instance, the algorithm by Fomin et al.~\cite{fomin2007exact} solves $H$-COLORING in  $O^*(( t+3 )^n)$ time if $H$ has tree-width at most $t$ on every graph $G$ with $|V_G|=n$, assuming that the tree-decomposition is given. Moreover, Wahlström \cite{wahlstrom2011new} designed an algorithm based on $k$-expressions and clique-width in order to solve the $\#H$-COLORING problem in time $O^*( (2k+1)^n )$, assuming that the clique-width of the graph $H$ is $k\geq 1$. The running time descends to $O^*((k+2)^n)$ if $H$ admits a linear $k$-expression, leading to a $O^*(6^n)$ complexity for $\#C_p$-COLORING (where $C_p$ is the $p$-cycle for $p\geq 5$, $C_p$ having a linear $4$-expression), and an $O^*(5^n)$ complexity for $\#H$-COLORING, for any cograph $H$ (the cographs with at least one edge being exactly the graphs of clique-width $2$). Recently, Okrasa and Rz\c ażewski gave in \cite{okrasa2020finegrained} an algorithm running in time $O^*( |V_H|^{tw(G)} )$ solving $H$-COLORING (assuming an optimal tree-decomposition of $G$ is given) and established its optimality, in the sense that $H$-COLORING can not be solved in time $O^*( (|V_H|-\varepsilon)^{tw(G)} )$ for all $\varepsilon>0$ (under the hypothesis that $H$ is a projective core), unless the {\em strong exponential-time hypothesis} (SETH) fails\footnote{I.e., that SAT can not be solved in time $(2-\varepsilon)^n\times (n+m)^{O(1)}$ on instances with $n$ variables and $m$ clauses. }. They also achieved similar results involving clique-width \cite{DBLP:conf/icalp/GanianHKOS22}, where they showed the existence and the optimality of their algorithm running in $O^*(s(H)^{cw(G)})$ (with $s(H)$ being a new structural parameter of $H$), in the sense that the running time $O^*((s(H)-\varepsilon)^{cw(G)})$ can not be reached under the SETH.

In this paper we pursue this line of research and study the fine-grained and parameterized complexities of $(\#)H$-COLORING problems and, more generally, of \#BINARY-CSP$(\Gamma)$ with a particular focus on the parameter {\em twin-width}, recently introduced and now widely investigated~\cite{DBLP:conf/icalp/BergeBD22,bonnet2022twin8,bonnet2022twin-exponential-treewidth,bonnet2020twin3,bonnet2021twin2,bonnet2022twin7,bonnet2022twin4,bonnet2022twin6,bonnet2022twin-poly-kernels,bonnet2020twin1,bonnet2022reduced,bonnet2021twin-permutations}. This parameter, as well as the underlying {\em contraction sequences}, give information on whether two vertices are ``similar'' (in the sense that they have almost the same neighborhoods) in order to treat them simultaneously and thus reduce the computation time. A major achievement of twin-width is that deciding whether a graph $G$ is a model of a  closed first-order formula $\varphi$ is FPT when parameterized by the twin-width of $G$ and the length of $\varphi$ \cite{bonnet2020twin1}. Also, the $k$-INDEPENDENT SET problem is FPT when  parameterized by $k$ and twin-width \cite{bonnet2020twin3}. In \cite{bonnet2022twin6} Bonnet et al. also gave a proof that the $q$-COLORING problem was FPT when  parameterized by {\em component twin-width}. Specifically, the parameter {\em component twin-width} appears to be the most relevant and natural parameter in the setting of homomorphism problems. In fact,  many of the algorithms in~\cite{bonnet2020twin3} that are FPT when parameterized by twin-width, such as the one solving $k$-IND-SET, can be seen as an optimization of a more natural FPT algorithm parameterized by component twin-width.

Despite their impressive success,  parameters based on twin-width
have not been used to tackle $\#H$-COLORING problems, except for the very restricted case of $q$-COLORING briefly discussed in \cite{bonnet2022twin6}.
To extend the applicability of such parameters to larger classes of graphs and problems, in Section \ref{sec:edge_labelled_graphs} we propose a representation of instances and templates of BINARY-CSPs as ``edge-labelled graphs''. However,  classical representations via Gaifman graphs are not adapted to algorithms dealing with component twin-width because the latter does not necessarily increase with higher number of edges.

We thus propose the algebraic notion of {\em pre-morphism into a semiring} in Section \ref{sec:semirings unification}, that enables a unified formalism for the list, the counting and the cost generalisations of BINARY-CSP. This formalism also subsumes the semiring based  generalisations of CSP (SCSP) by Bistarelly \cite{bistarelli1997semiring}, including in particular the counting version of CSP. Even though Wilson \cite{DBLP:conf/ijcai/Wilson05} proposed an equivalent generalisation through semirings, Wilson's formalism can be seen as a descendant approach, studying subsets of the set of solutions with less and less specifications over time. In contrast, our method builds the set of solutions from trivial cases and two basic operations $\uplus$ and $\Join$, that we will perform (through our pre-morphism) in the semiring instead. This approach is often a more common and prolific way to involve graph parameters; see, e.g.,  Wahlström \cite{wahlstrom2011new} in the context of clique-width.

Finally, in Section \ref{sec:tractability ctww} we use contraction sequences along with component twin-width to implement dynamic programming algorithms that  efficiently solve these generalisations of BINARY-CSPs. Here, we consider both the case where we bound the input graphs and the case where we bound the template by component twin-width, and we solve both of these questions by two novel algorithms: an FPT algorithm applicable for inputs of bounded component twin-width, and a superpolynomial but significantly improved algorithm applicable to templates of bounded component twin-width.
They strongly generalize and improve the results by, e.g., Wahlstr\"om \cite{wahlstrom2011new}, and are to the best of our knowledge the most general algorithms of their kind (for binary constraints). In fact, our two algorithms even solve combinations of generalisations of BINARY-CSP without impacting the running time. In Table \ref{tab:upper bounds} we summarize a few cases where our approach improves the upper bounds that we can derive from  \cite{wahlstrom2011new}, which uses $k$-expressions and clique-width, and which complements the $q$-COLORING problems untreated by the inclusion-exclusion method of Bj\"orklund et al.~\cite{bjorklund2009set}. Regarding the latter, we also observe that while it runs in $O^*(2^n)$ time and is able to solve many combinations of generalisations of $q$-COLORING, it does not cover all combinations of generalisations of BINARY-CSP solved by our algorithms, and is naturally restricted to the very specific case of complete graphs. 
For our FPT algorithms, component twin-width and clique-width are functionally equivalent \cite{bonnet2022twin6} on graphs, which allows us to derive FPT conditions for generalized $H$-COLORING problems with respect to clique-width.
This extends the corollary derived from Courcelle et al. \cite{courcelle1990monadic} that for every graph $H$, $H$-COLORING is FPT when parameterized by clique-width even if we add costs, since we now allow combinations of counting generalisations. We summarize the FPT results in Table \ref{tab:tractability results}.
These results raise several questions for future research, and we discuss some of them in Section~\ref{sec:conclusions}.

\begin{table}
    \centering
    \begin{tabular}{c c c c c c}
    \textbf{Generalisation$\setminus$Graph} & Clique & Cograph & Even  cycle & Odd cycle \\
    \hline
    $H$-COLORING & $\textcolor{blue}{O^*(3^n)}$ & $\textcolor{blue}{O^*(3^n)}$ & $\textcolor{green}{O^*(5^n)}$ & $\textcolor{violet}{O^*(5^n)}$ \\
    $\#H$-COLORING & $\textcolor{blue}{O^*(3^n)}$ & $\textcolor{red}{O^*(3^n)}$ & $\textcolor{red}{O^*(5^n)}$ & $\textcolor{red}{O^*(5^n)}$ \\
    list-$H$-COLORING & $\textcolor{blue}{O^*(3^n)}$ & $O^*(3^n)$ & $O^*(5^n)$ & $O^*(5^n)$  \\
    \#list-$H$-COLORING & $\textcolor{blue}{O^*(3^n)}$ & $O^*(3^n)$ & $O^*(5^n)$ & $O^*(5^n)$  \\
    cost-$H$-COLORING & $\textcolor{blue}{O^*(3^n)}$ & $O^*(3^n)$ & $O^*(5^n)$ & $O^*(5^n)$  \\
    \#cost-$H$-COLORING & $O^*(3^n)$ & $O^*(3^n)$ & $O^*(5^n)$ & $O^*(5^n)$  \\
    \#list-cost-$H$-COLORING & $O^*(3^n)$ & $O^*(3^n)$ & $O^*(5^n)$ & $O^*(5^n)$  \\
    weighted-$H$-COLORING & $\textcolor{blue}{O^*(3^n)}$ & $O^*(3^n)$ & $O^*(5^n)$ & $O^*(5^n)$  \\
    \#weighted-$H$-COLORING & $O^*(3^n)$ & $O^*(3^n)$ & $O^*(5^n)$ & $O^*(5^n)$  \\
    \#list-weighted-$H$-COLORING & $O^*(3^n)$ & $O^*(3^n)$ & $O^*(5^n)$ & $O^*(5^n)$  \\    
    restricted-$H$-COLORING & $\textcolor{blue}{O^*(3^n)}$ & $O^*(3^n)$ & $O^*(5^n)$ & $O^*(5^n)$  \\
    \#restricted-$H$-COLORING & $\textcolor{blue}{O^*(3^n)}$ & $O^*(3^n)$ & $O^*(5^n)$ & $O^*(5^n)$  \\
    \#list-restricted-$H$-COLORING & $\textcolor{blue}{O^*(3^n)}$ & $O^*(3^n)$ & $O^*(5^n)$ & $O^*(5^n)$  \\
    \#list-restricted-cost-$H$-COLORING & $O^*(3^n)$ & $O^*(3^n)$ & $O^*(5^n)$ & $O^*(5^n)$ \\
    \#list-restricted-weighted-$H$-COLORING & $O^*(3^n)$ & $O^*(3^n)$ & $O^*(5^n)$ & $O^*(5^n)$  \\
    \hline
    \end{tabular}
    \caption{Upper bounds derived from Theorem \ref{thm:Solving Omega H-(R-MORPHISM)} and Algorithm \ref{algo:main algorithm}. The $O^*(3^n)$ bounds in blue can be improved to $O^*(2^n)$ by the inclusion-exclusion method \cite{bjorklund2009set}. Even-cycle coloring is in P, (our algorithm solves it in $O^*(5^n)$, in green on this table). Also, $C_{2k+1}$-COLORING (with $C_{2k+1}$ the $(2k+1)$ odd-cycle) can be done in $O^*((\alpha_k)^n)$, with $(\alpha_k)_{k\geq 1}$ decreasing and tending to $1$, and with $\alpha_1\leq\sqrt{2}$ \cite{fomin2007exact} (our algorithm solves it in $O^*(5^n)$, in violet in this table). We improve the results of Wahlström \cite{wahlstrom2011new} of $\#$cograph-COLORING from $O^*(5^n)$ to $O^*(3^n)$, and of $\#$cycle-COLORING from $O^*(6^n)$ to $O^*(5^n)$ (in red), while also generalising them. To our knowledge, no fine-grained algorithm have been given in the literature for every other problem in this table.}
    \label{tab:upper bounds}
\end{table}

\section{Preliminaries}
\label{sec:prel}

In this section we recall the basic notation and terminology that will be used throughout the paper.

\subsection{Basic Notation}

A {\em graph} $H$ is a tuple $(V_H,E_H)$ where $V_H$ is a finite set referred as the {\em set of vertices of }$H$ and $E_H$ is a binary relation over $V_H$, called the {\em set of edges of }$H$\footnote{Any graph $H$ will always be denoted $H=(V_H,E_H)$.}. The graph $H$ is said to be {\em loopless} if $E_H$ is irreflexive and {\em non-oriented} if $E_H$ is symmetric.

Given an arbitrary set $A$, we denote by $\P(A)$ the set of all subsets of $A$. For $S\in \P(\P(A))$ (i.e. $S\subseteq \P(A)$), let $\cup S$ denote the set ${\scriptstyle \bigcup\limits_{s\in S}} s$. Note that $(\cup S\in \P(A))$ (i.e.\ $S\subseteq A$). We say that $S$ is a {\em proper partition of }$A$, if every $a\in A$ belong to a unique $s\in S$, and that $\emptyset\notin S$. For two proper partitions $S'$ and $S$ of $A$, we say that $S'$ {\em respects} $S$ if for all $s'\in S'$, there exists $s\in S$ such that $s'\subseteq s$. We occasionally relax the notation and allow a proper partition of $A$ to be family of pairwise disjointed non-empty subsets of $A$ whose union equals $A$. This will be useful in contexts where the order of the partition matters. If $B$ is an arbitrary set, we let $B^A$ be the set of functions from $A$ to $B$.
For $f\in B^A$ and $A'\subseteq A$, the restriction of $f$ to $A'$ is the function $f|_{A'}: A' \mapsto  B$. Also, for $B'\subseteq B$ and $f\in B^A$ with $f(A)\subseteq B'$, the corestriction of $f$ to $B'$ is the function $f|^{B'}: A \mapsto  B'$

Throughout, $p$ will only be used to denote an integer $\geq 1$, and $[p]$ is the set $\{1,\dots,p\}$. We let $\N$, $\mathbb{R}$, $\mathbb{R}_+$ denote the natural, real, and positive real numbers, respectively. We let $\overline{\N}=\N\cup\{+\infty\}$,  $\Rbar:=\mathbb{R}\cup\{+\infty\}$ and $\Rbar_+:=\mathbb{R}_+\cup\{+\infty\}$. If $A\subseteq \Rbar$, we denote its minimum by $\min A$, with the convention that $\min \emptyset = \infty$. Similarly, if $A\subseteq \Rbar_+$, we denote its maximum by $\max A$, with the convention that $\max \emptyset = 0$. Also, for $n,m\geq 1$, $a=(a_1,\dots,a_n)\in A^n$, and $a'=(a_{n+1},\dots,a_{n+m})\in A^m$, we denote the concatenation of $a$ and $a'$  by $a,a'$ and $(a,a')$ that is defined as the $(n+m)$-tuple $(a_1,\dots,a_n,a_{n+1},\dots,a_{n+m})\in A^{n+m}$. For $I=\{i_1,\dots i_{|I|}\}\subseteq [n]$ with $i_1<\dots<i_{|I|}$, we denote by $a_I$ and $(a)_I$  the tuple $(a_{i_1},\dots,a_{i_{|I|}})\in A^{|I|}$.

\subsection{Constraint Satisfaction Problems}

The {\em constraint satisfaction problem} asks whether it is possible to assign values to variables while satisfying all given constraint of the instance. Additionally, it is common to parameterize the problem by a set of relations $\Gamma$ and a domain $D$, and constraints are then only allowed to use relations from $\Gamma$.

\begin{algorithm}

\textbf{CSP}$(\Gamma)$:

\textbf{Input:} A set $V$ of variables and a set $C$ of constraints of the form $(R_i,(v^1_i,\dots,v^{\text{ar}(R_i)}_i))$ where $R_i\in\Gamma$ and $(v^1_i,\dots,v^{\text{ar}(R_i)}_i)\in V^{\text{ar}(R_i)}$ (where $\text{ar}(R_i)$ is the arity of the relation $R_i$).

\textbf{Output:} 1 if there exists a function $f:V\mapsto D$ such that for all constraints $(R_i,(v^1_i,\dots,v^{\text{ar}(R_i)}_i))\in C$, we have $(f(v^1_i),\dots,f(v^{\text{ar}(R_i)}_i))\in R_i$, 0 otherwise.

\end{algorithm}

If $\Gamma$ only contains binary relations, then we write BINARY-CSP$(\Gamma)$.
For example, when $\Gamma=\{ E_H \}$,  BINARY-CSP($\Gamma$) is equivalent to the well-known {\em (digraph) $H$-COLORING problem}, often referred as $H$-COLORING when $H$ is symmetric. Regarding the complexity of the latter, Hell and Ne{\v{s}}et{\v{r}}il \cite{hell1990complexity} proved that $H$-COLORING is in P when  $H$ is bipartite, and  it is NP-complete, otherwise. This result was later generalized to the CSP dichotomy theorem, proven independently by Zhuk \cite{DBLP:journals/jacm/Zhuk20} and Bulatov~\cite{DBLP:journals/corr/abs-2007-09099}, which states that every problem of the form CSP($\Gamma$) is either in P or it is NP-complete. In the sequel, we prefer to state our algorithmic results for the most general problem possible (typically BINARY-CSP$(\Gamma)$ since they imply the results of the specific problems (e.g., $H$-COLORING). Additionally, we consider the following generalized problems:

\begin{itemize}

\item The {\em counting} version (\#):  how many functions $f:V\mapsto D$ are solutions of the instances?

\item The {\em list} version:
given a matrix $L=\{0,1\}^{V\times D}$ and an instance $\mathcal{I}$, is there  a solution $f:V\mapsto D$ satisfying $\forall (u,d)\in V\times D, f(u)=d \implies L(u,d)=1$?

\item The {\em cost} version \cite{kobler2003edge}:
given a matrix $C\in\Rbar^{V\times D}$ and an instance $\mathcal{I}$, what is the minimum value of ${\scriptstyle\sum \limits_{u\in V}} C(u,f(u))$ for  a solution $f:V\mapsto D$?

\item The {\em weighted} version \cite{escoffier2006weighted}:
given a matrix $W\in\Rbar_+^{V\times D}$ and an instance $\mathcal{I}$, what is the minimum value of
${\scriptstyle\sum\limits_{d\in D}} \max\limits_{u\in f^{-1}(\{d\})} W(u,f(u))$ for a solution $f:V\mapsto D$? 

\item The {\em restrictive} version \cite{diaz2005restrictive}: given a tuple $R=\Nbar^D$ and and instance $\mathcal{I}$, does there exist a solution $f:V\mapsto D$ where $\forall d\in D, R(d)\neq +\infty \implies |f^{-1}(\{d\})|=R(d)$?
\end{itemize}

These generalisations can naturally also be combined together. For example, the counting-cost version is the task of counting solutions  of minimal cost.
We will see in Section~\ref{sec:semirings unification} that the semiring formulation of CSP makes it possible to subsume all of these generalizations into a single formalism.

\subsection{Parameterized Complexity}

For a computational problem $\mathcal{Q}$ we let $dom(\mathcal{Q})$ be the set of all possible instances of $\mathcal{Q}$. A function $\mu:dom(\mathcal{Q}) \mapsto \N$ is called a {\em parameter} of $\mathcal{Q}$. We say that $\mathcal{Q}$ is {\em fixed-parameter tractable w.r.t. $\mu$} if there exists an algorithm solving $\mathcal{Q}$ on every instance $x\in dom(\mathcal{Q})$ of size $\|x\|$ in time $O(f(\mu(x))\times \|x\|^{O(1)})$, where $f$ can be any computable function. For additional details, we refer to the textbook by Downey and Fellows~\cite{downey2012parameterized}.
For example, {\em treewidth} of graphs is a well-known parameter frequently resulting in fixed parameter tractability, including the (\#)$H$-COLORING problem~\cite{diaz2002counting}.
Unfortunately, computing treewidth is in general NP-hard, and the classes of graphs of bounded tree-width are rare. Hence,  finding additional examples of (graph width) parameters can in practice be very useful.
Another frequently occurring parameter is {\em clique-width}. This parameter can be seen as a generalisation of tree-width, in the sense that any treewidth-bounded class of graphs is also clique-width bounded, even though the reverse is not true. The use of clique-width has led to many positive results in graph algorithms, e.g., Bodlaender et al.~\cite{bodlaender2010faster} designed an algorithm solving Min-DOMINATING-SET in time $O^*(3^{\frac{\omega}{2}k})$ on graphs of clique width $\leq k$ ($\omega < 2.376$). Kobler et al. \cite{kobler2003edge} also solved a cost version of list-$q$-COLORING in time $O(2^{2qk}qk^3n)$ on graphs $G$ with $n$ vertices and clique-width $k$, and more generally Wahlstr\"om proved~\cite{wahlstrom2011new} that (\#)$H$-COLORING is FPT when parameterized by the clique-width of the input graph.
Recently, Bonnet et. al \cite{bonnet2020twin1} introduced {\it contraction sequences} in order to use dynamic programming on graphs. Contraction sequences allows one to derive a lot of additional parameters such as {\em twin-width}, oriented twin-width, total twin-width and component twin-width \cite{bonnet2022twin6}.
We do not formally define these parameters here since we in Section~\ref{sec:edge_labelled_graphs} define these in a slightly more general context, but remark that component twin-width and clique-width (and thus rank-width) have been proven to be functionally equivalent, allowing one to translate FPT results between these parameters. 

\subsection{Fine-Grained Complexity}

A related approach to tackle (NP-)hard problem is the construction of exponential time algorithms running in $O^*(c^n)$ time  for as small $c$ as possible, where $n$ is a complexity parameter. In this paper, $n$ will always refer to the number of variables (or vertices) of a CSP (or $H$-COLORING) instance. This approach, especially in the context of proving matching {\em lower bounds} under the {\em (strong) exponential-time hypothesis} ((S)ETH), is often called {\em fine-grained complexity}.

Notable, general results for $H$-COLORING problems include the $O^*((2k + 1)n)$ algorithm by Fomin et al.~\cite{fomin2007exact} for $H$-COLORING for graphs $H$ of treewidth $k$, and
Wahlstr\"om's~\cite{wahlstrom2011new} algorithm based on $k$-expressions and clique-width which solves $\#H$-COLORING problem in time $O^*((2k+1)^n )$  when $H$ has clique-width $k\geq 1$. The running time descends to $O^*((k+2)^n)$ if $H$ admits a linear $k$-expression, leading to a $O^*(6^n)$ complexity for $\#C_q$-COLORING (where $C_q$ is the $q$-cycle for $q\geq 5$, $C_q$ having a linear $4$-expression), and a $O^*(5^n)$ complexity for $\#H$-COLORING, for any cograph $H$. As we will see later, these bounds can in many cases be significantly improved by considering component twin-width instead of clique-width.

\section{Edge-Labelled Graphs and $\R$-morphisms} \label{sec:edge_labelled_graphs}

In this section we consider a generalization of graphs, {\em edge-labelled graphs}, where the corresponding morphism notion greatly increases the expressive power and leads to a rich computational problem. The main idea is that the morphism notion encodes the ``rules'' of the problems: for example, the ``rule'' respected by homomorphisms are that edges are sent to edges, whereas in the subgraph-isomorphism problem, the ``rules'' are that edges are sent to edges, non-edges to non-edges, and different vertices are sent to different vertices. 
We begin in Section~\ref{sec:defs} by extending the concept of graphs by allowing labels on every pair of vertices, and explicit the notion of morphism relatively to a binary relation over the sets of labels, in which every pair of vertices must be send to a pair of vertices whose label is prescribed by the relation. Then, in Section \ref{sec:hom_equivalence}, we prove that the associated morphism problems naturally encode exactly the problems of the form BINARY-CSP($\Gamma$) via a reduction which does not introduce any fresh variables, and where the sets of solutions does not change. Last, in Section~\ref{sec:contractions} we show how to formulate the component twin-width parameter in the context of edge-labelled graphs.

\subsection{Definitions} \label{sec:defs}

We first introduce the concept of ``edge-labelled graph''.

\begin{definition} \label{def:el_graph}
An {\em edge-labelled graph} $G$ is a structure $G=(V_G,l_G,X_G)$ where $V_G$ and $X_G$ are finite sets and $l_G:(V_G)^2\mapsto X_G$. The sets $V_G$ and $X_G$ will be referred to respectively the sets of {\em vertices} and {\em labels of edges} of $G$, and $l_G$ will be referred as the {\em label function} of edges of $G$.
\end{definition}

For any edge-labelled graph $G$, the set of vertices of $G$, the set labels of edges of $G$, and the label function of edges of $G$, will always be denoted $V_G$, $X_G$ and $l_G$ (as in Definition~\ref{def:el_graph}). For $S\subseteq V_G$, $G[S]$ denotes the edge-labelled graph induced by $S$ on $G$, i.e., $G[S]:=(S,l_G|_{S^2},X_G)$.

We now introduce a symbol $\red$ that will play a special role as a label for edge-labelled graph (see Section \ref{sec:contractions} for more details). We implicitly assume that this symbol does not occur in any other context.
An edge-labelled graph $G$ is said to be {\em $\red$-free} if $\red\notin X_G$.
We now have the necessary technical machinery to properly generalize the concept of a graph homomorphism to edge-labelled graphs. 

\begin{definition}

Let $G$ and $H$ be two $\red$-free edge-labelled graphs, and let $\R\subseteq X_G\times X_H$.
A function $f \colon V_G\mapsto V_H$ is said to be an $\R${\em -morphism} ($f \colon G\underset{\R}{\rightarrow} H$) if
$\forall (u,v)\in (V_G)^2, (l_G(u,v),l_H(f(u),f(v))) \in \R$.

\end{definition}

In this definition, the relation $\R$ encodes whether a pair of vertices $(u,v)$ of $G$ of label $x\in X_G$ is allowed to be sent to a pair of vertices $(a,b)$ of $H$ of a label $y\in X_H$: which it is if and only if $(x,y)\in\R$.
By a slight abuse of notation, viewing a graph as an edge-labelled graph whose edges are labelled by $1$, and every other pair of vertices is labelled by $0$, we notice that a homomorphism is exactly a HOM-morphism if we let HOM be the binary relation over $\{0,1\}$ defined as HOM$=\{(0,0),(0,1),(1,1)\}$.

Similarly to $H$-COLORING, we will mostly be interested in the version of this problem where the target edge-labelled graph $H$ is fixed. This allows us to model different types of problems simply by changing the template $H$.
Thus, let $H$ be a $\red$-free edge-labelled graph, $X$ a finite set, and $\mathcal{R}\subseteq X\times X_H$. We define the following computational problem.

\begin{algorithm}

\textbf{$H$-($\mathcal{R}$-MORPHISM):}

\textbf{Instance:} An $\red$-free edge-labelled graphs $G$ with $X_G=X$.

\textbf{Question:} Does there exist a function $f \colon G \underset{\mathcal{R}}{\rightarrow} H$ ?

\end{algorithm}

Note that $H$-COLORING is the same problem as $H$-(HOM-MORPHISM), i.e., $H$-($\mathcal{R}$-MORPHISM) is at least as expressive as $H$-COLORING. We will see Section \ref{sec:hom_equivalence} that these problems encode exactly the various BINARY-CSP($\Gamma$) problems, for any set of binary relations $\Gamma$. 

\subsection{Equivalence Between $H$-($\mathcal{R}$-MORPHISM) and BINARY-CSP$(\Gamma)$} \label{sec:hom_equivalence}

We now show that any BINARY-CSP($\Gamma$) problem can be reformulated as a $H$-($\mathcal{R}$-MORPHISM) problem. For any set $\Gamma$ of binary relation over a finite domain, we thus need to build an $\red$-free edge-labelled graph $H(\Gamma)$, a set $X$ and a relation $\R_{\Gamma}\subseteq X\times X_{H(\Gamma)}$, in such a way that BINARY-CSP($\Gamma$) and $H(\Gamma)$-($\mathcal{R}_{\Gamma}$-MORPHISM) are the same problems.

\begin{definition}

For a finite domain $D$ and a set of binary relations $\Gamma$ over $D$ we let:

\begin{itemize}

    \item $H(\Gamma):=(D,l_{\Gamma},\mathcal{P}(\Gamma))$ with $l_{\Gamma}$ being defined by: $\forall (a,b)\in D, l_{\Gamma}(a,b) = \{ R\in\Gamma\mid (a,b)\in R\}$, and
    
    \item $\mathcal{R}_{\Gamma}$ be defined by, for all $(Y,Z)\in\mathcal{P}(\Gamma)^2, (Y,Z)\in\mathcal{R}_{\Gamma} \iff Y\subseteq Z$.
    
\end{itemize}

\end{definition}

Next, we show how to build an equivalent instance of $H(\Gamma)$-($\R_{\Gamma}$-MORPHISM), given an instance of BINARY-CSP($\Gamma$).
Intuitively, for all $(u,v)\in V^2$, the label $l_{G(\mathcal{I})}(u,v)$ represent the constraints that must be respected for the tuple $(u,v)$, while, for $f \colon V\mapsto D$, the label $l_{\Gamma}((f(u),f(v)))$ represent the constraints of $\Gamma$ that are actually satisfied by $(f(u),f(v))$.
In the following theorem, for an instance $\mathcal{I} = (V,C)$ of BINARY-CSP($\Gamma$), we let $G(\mathcal{I}):=(V,l_{\mathcal{I}},\mathcal{P}(\Gamma))$ with $l_{\mathcal{I}}$ being defined by: $\forall (u,v)\in V, l_{\mathcal{I}}(u,v) := \{ R\in\Gamma\mid R(u,v) \in C\}$.

\begin{restatable}{theorem}{BinaryCSPtoRmorphism}\label{thm:Binary-CSP to R-morphism}
Let $\Gamma$ be a set of binary relations over a finite domain $D$, and let $\mathcal{I}$ be an instance of BINARY-CSP$(\Gamma)$ over a set of variables $V$.
Then, for every $f \colon V\mapsto D$, $f$ is a solution to $\mathcal{I}$ if and only if $f$ is a solution to the instance $G(\mathcal{I})$ of $H(\Gamma)$-$(\R_{\Gamma}$-MORPHISM$)$.
\end{restatable}

Conversely, we show that any $H$-($\mathcal{R}$-MORPHISM) problem can be reformulated as a BINARY-CSP via the following translation.

\begin{definition}
Let $H=(V_H,l_H,X_H)$ be an $\red$-free edge-labelled graph, $X$ a finite set and $\mathcal{R}\subseteq X\times X_H$.
For all $x\in X$, define the binary relation over $V_H$: $R_x:= \{ (a,b) \in (V_H)^2 \mid (x,l_H(a,b)) \in \mathcal{R} \}$, and let $\Gamma_{H,\R} := \{ R_x, x\in X \}$. 
\end{definition}

In the following theorem, for two edge-labelled graphs $G$ and $H$ and $\R\subseteq X_G\times X_H$, we let $\mathcal{I}(G)$ be the instance of BINARY-CSP$(\Gamma_{H,\R})$ with variables $V_G$ and constraints $\{R_{l_G(u,v)}(u,v) \mid (u,v) \in (V_G)^2)\}$.

\begin{restatable}{theorem}{RmorphismtoBinaryCSP}\label{thm:R-morphism to Binary-CSP}

Let $H=(V_H,l_H,X_H)$ be an $\red$-free edge-labelled graph, $X$ a finite set and $\mathcal{R}\subseteq X\times X_H$.
Let $G$ be an instance of $H$-($\mathcal{R}$-morphism), i.e. an $\red$-free edge-labelled graph with $X_G=X$.
Then, for any $f \colon V_G\mapsto V_H$, $f$ is a solution to the instance $\mathcal{I}$ of BINARY-CSP$(\Gamma)$ if and only if $f$ is a solution to the instance $G$ of $H$-($\mathcal{R}$-MORPHISM).

\end{restatable}

We summarize the translation between the vocabulary of BINARY-CSP and the equivalent concepts of $H$-($\R$-MORPHISM) in Table~\ref{tab:translation CSP,MORPHISM}. We again remind the reader that this formulation of BINARY-CSP in terms of edge-labelled graphs makes it easier to employ graph parameters such as the aforementioned component twin-width parameter~\cite{bonnet2022twin6}, which we will explicitly show in the forthcoming section.

\begin{table}
\centering
\begin{tabular}{ c  c  c }
    \textbf{BINARY-CSP($\Gamma$)} & \textbf{$H$-($\R$-MORPHISM)} \\
    \hline
    $V$ (set of variables) & $V_G$ (vertices of $G$) \\
    $D$ (domain) & $V_H$ (vertices of $H$) \\
    $\Gamma$ & $\R$, $X$ and $H$ \\
    Constraint(s) required for $(u,v)\in V^2$ & $l_G(u,v)\in X$ \\
    Constraint(s) respected by $(a,b)\in D^2$ & $l_H(a,b)\in X_H$ \\
    $f:V\mapsto D$ respects every contraint & $f:V_G\mapsto V_H$ is an $\R$-morphism \\
    Graph homomorphism & $\text{HOM}$-morphism with $\text{HOM}:=\{(0,0),(0,1),(1,1)\}$ \\
    $H$-COLORING & $H$-($\text{HOM}$-MORPHISM) \\
    \hline
\end{tabular}
\label{tab:translation CSP,MORPHISM}
\caption{Translations between BINARY-CSP and the edge-labelled graph formalism. We view graphs as edge-labelled graphs where the edges are labelled by $1$, and every other pair of vertices is labelled by $0$.}
\end{table}

\subsection{Contractions of Edge-Labelled Graphs and Component Twin-Width}
\label{sec:contractions}

We now generalize the notion of graph contraction defined by Bonnet et al.~\cite{bonnet2020twin1} to edge-labelled graphs, and define the key notions of contraction sequences and component twin-width.
We first define a suitable notion of vertex merging in an edge-labelled graph.

\begin{definition}\label{def:Meaning Contraction}
Let $H=(V_H,l_H,X_H)$ be an edge-labelled graph, and let $\S$ be a proper partition of $V_H$.
The {\em contraction of $H$ relative to $\S$} is the edge-labelled graph $H_{\S}=(V_{H_{\S}},l_{H_{\S}},X_{H_{\S}})$ with $V_{H_{\S}} = \S$, $X_{H_{\S}} = X_H\cup \{\red\}$ and for all $(S_1,S_2)\in \S^2$, if there exists $x\in X_H$ such that $\forall (u,v)\in S_1\times S_2$, $l_G(u,v) = x$, then $l_{H_{\S}}(S_1,S_2) = x$, and $l_H(S_1, S_2) = \red$ otherwise.

\end{definition}

If we view a graph as an edge-labelled graph whose edges are labelled by $1$, and every other pair of vertices is labelled by $0$, this notion coincides with the notion of (iterated) contraction(s) defined by Bonnet et al.~\cite{bonnet2020twin1} for graphs, up to identifying the set of red edges (defined in \cite{bonnet2020twin1}) with the set of pairs of $\S$ labelled with $\red$.
With a slight abuse of notation, we see a proper partition $\S_1$ of a proper partition $\S_2$ of a set $V$ as a proper partition of $V$, by seeing any set $S\in\P(\P(V))$ of $\S_1$ as $\cup S \in\P(V)$. Under this viewpoint, the proper partition $\S_1$ of $V$ respects the proper partition $\S_2$ of $V$, and the contraction relation is then transitive. This naturally leads to the following definition of a {\em contraction sequence}.

\begin{definition}

Let $H$ be an $\red$-free edge-labelled graph on $n\geq 1$ vertices.
Then, a {\em contraction sequence} of $H$ is a sequence of edge-labelled graph $(H_n,\dots,H_1)$ such that $H_n=H$, and for all $k\in [n-1]$, $H_k$ is a contraction of $H_{k+1}$ with $|V_{H_k}|=k$.

\end{definition}

In particular, $H_1$ is an edge-labelled graph with 1 vertex, and $V_{H_1} = \{V_H\}$.
Again, one may notice that the notion of contraction sequence of edge-labelled graph coincides with the particular case of binary loopless non-oriented graphs defined in~\cite{bonnet2020twin1}.

\begin{definition}
Let $H$ be an edge-labelled graph. Define the {\em $\red$-connected components} of $H$ as the connected components of the unoriented graph $H(\red)=(V_H, \{ (u,v)\in (V_H)^2 \mid \red\in \{l_H(u,v),l_H(v,u)\}\})$.
\end{definition}

The edges of $H(\red)$ encode a loss of information in the contraction, and are intended to play the role of the red edges of contraction sequences defined in~\cite{bonnet2020twin1}.
We can now introduce the notion of component (twin-)width of a contraction sequence

\begin{definition}
Let $H$ be an $\red$-free edge-labelled graph on $n\geq 1$ vertices. Let $(H_n,\dots,H_1)$ be a contraction sequence of $H$. The {\em component-width} of this contraction sequence is denoted  by $ctw((H_n,\dots H_1))$ and is the maximal size of the $\red$-connected component of the graphs $H_{n-1},\dots,H_1$. The {\em component twin-width} of $H$ is denoted by $ctww(H)$, and it  is the minimal component-width of all its contraction sequences. Any contraction sequence of $H$ whose component-width equals the component twin-width of $H$ is called an {\em optimal contraction sequence} of $H$.
\end{definition}

Note that the graphs of component twin-width $1$ are exactly the cographs, and that cycles of length $\geq 5$ have component twin-width $3$.
In order to extend the notion of component twin-width to BINARY-CSP$(\Gamma)$ compatible with the reduction involved in Theorem \ref{thm:Binary-CSP to R-morphism}, we define the {\em component twin-width} of the template $\Gamma$ as the component twin-width of $H(\Gamma)$, and the instance $\mathcal{I}$ as the component twin-width of $G(\mathcal{I})$.
Note that while the primary focus of~\cite{bonnet2020twin1} was simply the {\em twin-width} (the minimum over all contraction sequences of the maximal $\red$-degree of the graphs that appear in the contraction sequence), the latter seems a less useful parameter for BINARY-CSP$(\Gamma)$. Intuitively,  one could argue that component twin-width is a more ``natural'' parameter to consider, but that improved algorithms using twin-width are sometimes feasible. For instance, Bonnet et al.~\cite{bonnet2020twin3} solves $k$-IND-SET problem in FPT when parameterized by twin-width, by  bounding the number of red-connected subgraphs by a function depending only on $k$ and $d$, where $d$ is the twin-width of the input graph. This trick transforms an algorithm parameterized by component twin-width into an algorithm parameterized by twin-width. Unfortunately, it does not seem  applicable here. Note also that to generalize algorithm of Bonnet et al.~\cite{bonnet2020twin3}   to IND-SET, we have to give up the parameterization by twin-width and instead consider component twin-width.

\section{Semirings and Generalisations of CSP}
\label{sec:semirings unification}

In this section we define and extend the {\em semiring} framework of CSP by Bisterelli~\cite{bistarelli1997semiring} and Wilson~\cite{DBLP:conf/ijcai/Wilson05} in the context of $H$-($\R$-MORPHISM) problems. In particular we will see that all the extensions of the basic CSP problem defined in Section~\ref{sec:prel} (e.g., counting, finding a solution of minimal cost) can be expressed within this framework, allowing all extended problems to be expressed within a single algebraic framework. 

\subsection{Semirings and Pre-Morphisms}

In this section, we will define a new algebraic notion in order to encompass the many generalisations of CSP evoked so far.

\begin{definition}
    A {\em semiring} is a structure $(A,+,\times,0_A,1_A)$ such that  $(A,+,0_A)$ is a commutative monoid, $(A,\times,1_A)$ is a monoid, $\times$ is distributive over $+$, and
        $0_A$ is absorbing for $\times$. Moreover, if $A$ is ordered by the binary relation $\leq_A$ over $A$ defined by $\forall (a,b)\in A^2, a\leq_A b\iff \exists c\in A, a+c=b$, then $(A,+,\times,0_A,1_A)$ is said to be a {\em dioid}.
\end{definition}

Note that rings and dioids are both particular cases of semirings.

\begin{definition}
    
    Let $S_1$ and $S_2$ be two disjointed sets, and $T_1$ and $T_2$ be two sets. Let $f_1\in (T_1)^{S_1}$ and $f_2\in (T_2)^{S_2}$.
    We define the {\em join} of $f_1$ and $f_2$ as $(f_1\Join f_2) \in (T_1\cup T_2)^{S_1\uplus S_2}$ defined by $(f_1\Join f_2)|_{S_1}=f_1$ and $(f_1\Join f_2)|_{S_2}=f_2$.
    Also, for $\F_1 \in \P((T_1)^{S_1})$ and $\F_2\in \P((T_2)^{S_2})$, we define the {\em join} of $\F_1$ and $\F_2$ by $\F_1\Join \F_2 = \{ (f_1\Join f_2), (f_1,f_2)\in\F_1\times\F_2 \} \in \P((T_1\cup T_2)^{S_1\uplus S_2}).$
    
\end{definition}

The join operation will be used by our algorithm, guided by the contraction sequences,  to iteratively extend the domain and codomain of the sets of functions considered, with the set of solutions being the final achievement of the algorithm. Similarly, we need disjointed union in order to extend the sets of functions considered. 
The basic idea behind the semiring framework is then to consider a set $A$ representing the sets of possible output of a function $\Omega$ applied to the set of solutions (of a given CSP instance). In the following definition, we will also take into account a weight matrix $W$, the weights being elements of a set $B$.

\begin{definition}
    Let $(A,+,\times,0_A,1_A)$ be a semiring, and $B$ a set.
    Let $\Omega$ be a function that maps, for $G$ and $H$ any two edge-labelled graphs, any couple $(\F,W)$ with $\F$ an element of $\P(T^S)$ (with $S\subseteq V_G$ and $T\subseteq V_H$), and $W\in B^{V_G\times V_H}$ a weight matrix, to an element of $A$, denoted by $\Omega_W(\F)$.\footnote{In fact, the value of $\Omega_W(\F)$ also depends on the graphs $G$ and $H$ considered, but to ease the notation they are omitted.}
        We say that $\Omega$ is a {\em $A$-pre-morphism with weights in $B$}\footnote{If the value of $\Omega_W(\F)$ does not depend on $W$, the precision the set $B$ is irrelevant, and we say that $\Omega$ is a $A$-{\em pre-morphism ignoring weights}.} if:
    
    \begin{itemize}
    
        \item For all edge-labelled graphs $G$ and $H$, $W\in B^{V_G\times V_H}$, $S\subseteq V_G$ and $T\subseteq V_H$:

        $\forall (\F_1,\F_2) \in (\P(T^S))^2, \F_1\cap \F_2 = \emptyset \implies \Omega_W(\F_1\uplus \F_2) = \Omega_W(\F_1) + \Omega_W(\F_2)$.

        \item For all edge-labelled graphs $G$ and $H$, and $W\in B^{V_G\times V_H}$, for all $S_1,S_2$ being two disjointed subsets of $V_G$ and $T_1,T_2$ being two disjointed subsets of $V_H$:
        
        $\forall (\F_1,\F_2) \in \P( (T_1)^{S_1} )\times \P( (T_2)^{S_2})$:
    $\Omega_W(\F_1\Join \F_2) = \Omega_W(\F_1) \times \Omega_W(\F_2)$.

        \item For all edge-labelled graphs $G$ and $H$ and $W\in B^{V_G\times V_H}$, $\Omega_W(\emptyset)=0_A$.

    \end{itemize}
    
\end{definition}

Additionally, if we can remove the assumption that $T_1$ and $T_2$ are disjointed in the second axiom then we say that the $A$-pre-morphism $\Omega$ is {\em strong}.
If the functions $+$ and $\times$ can be computed in constant time, and if $(G,H,W,S,a)\mapsto \Omega_W(\{f^S_{\{a\}}\})$ (with $f^S_{\{a\}}$ being the constant function of domain $S$ and codomain $\{a\}$) are polynomial time computable, we will say that $\Omega$ is a {\em poly-time computable} $A$-pre-morphism with weights in $B$.
Let us also remark that $\emptyset$ is neutral for $\uplus$, meaning that the third axiom stating that $\forall W, \Omega_W(\emptyset)=0_A$ can often be seen as a consequence of the first axiom in many practical cases. The presence of this axiom is only necessary to avoid pathological cases built especially to contradict this axiom that never occur in practice.
Similarly, denoting $f^{\emptyset}_{\emptyset}$ the unique function with an empty domain and an empty codomain, it is interesting to see that, in the cases that we consider, we will have that, for all weight matrix $W$, $\Omega_W(\{f^{\emptyset}_{\emptyset}\})=1_A$ as a consequence that $\{f^{\emptyset}_{\emptyset}\}$ is neutral for $\Join$ and of the second axiom.

In order to generalize the list, the counting, and the various weighted versions of $H$-($\R$-MORPHISM), consider the following problem, where 
$H$ is an edge-labelled graph $\R\subseteq X\times X_H$ with $X$ finite.

\begin{algorithm}
$\Omega$($H$-($\R$-\textbf{MORPHISM})):

\textbf{Input:} An instance $G$ of $H$-($\R$-\textbf{MORPHISM}), and $W\in B^{V_G\times V_H}$.

\textbf{Output:} The value of $\Omega_W(\{ f:G\underset{\R}{\rightarrow} H\})$.

\end{algorithm}

The CSP formulation ($\Omega$(CSP($\Gamma$))) is defined analogously, i.e., we ask for the output of $\Omega_W$ applied to the set of solutions to the input instance. As usual, we write ($\Omega$(BINARY-CSP($\Gamma$))) when $\Gamma$ is a set of binary relations. The main advantage of the $\Omega$($H$-($\R$-MORPHISM)) formulation is that it easily permits our algorithms to compute values of $\Omega$ of larger and larger sets of partial solutions of the given $H$-($\R$-MORPHISM) instance, via the two operations $\uplus$ and $\Join$.

\begin{remark}
The semiring generalisations of BINARY-CSP($\Gamma$), for  $\Gamma$ over a finite $D$,  subsume BINARY-CSP$(\Gamma)$, \#BINARY-CSP$(\Gamma)$, \#list-BINARY-CSP$(\Gamma)$ \#cost-list-BINARY-CSP$(\Gamma)$, \#weighted-list-BINARY-CSP$(\Gamma)$, and \#restricted-list-BINARY-CSP$(\Gamma)$, among others.  See Table~\ref{tab:semiring generalisations} for a summary of these problems and Appendix \ref{app:examples of pre-morphisms} for the precise definitions of the associated semirings and pre-morphisms.
Also, for a graph $H$, \#restrictive-list(BINARY-CSP($\{E_H\}$)) is the same problem as the counting version of the restrictive-list-$H$-COLORING problem defined by Díaz et. al in \cite{diaz2005restrictive}. Similarly, we could have implemented a weighted version of restrictive-list $H$-COLORING as a generalisation of $H$-COLORING, by taking a weight matrix $W$ with coefficient in $\overline{\N}\times\Rbar$, instead of simply $\Nbar\times\{0,1\}$.
\end{remark}

\begin{table}
\begin{tabular}{c c c c c}
     \textbf{Function $\Omega$} & $A$ & $B$ & $\Omega(H\text{-COLORING})$  & Target for Theorem \ref{thm:Solving Omega H-(R-MORPHISM) parameterized}\\ \hline
     $\mathbf{1}_{\neq\emptyset}$ & $\mathbf{2}$ & Unused & $H$-COLORING & YES \\
     list & $\mathbf{2}$ & $\mathbf{2}$ & list-$H$-COLORING & YES \\
     $\#$ & $\N$ & Unused & \#$H$-COLORING & YES \\
     $\#_{\text{list}}$ & $\N$ & $\mathbf{2}$ & $\#$list-$H$-COLORING & YES \\
     $\begin{matrix} \text{MinCost}, \\ \#\text{ArgMinCost} \end{matrix}$ & $\Rbar\times\N$ & $\Rbar$ & \#cost-list-$H$-COLORING \cite{kobler2003edge}  & YES \\
     $\begin{matrix} \text{MinWeight}, \\ \#\text{ArgMinWeight} \end{matrix}$ & $\Rbar\times\N$ & $\Rbar$ &  $\#$weighted-list-$H$-COLORING \cite{escoffier2006weighted}  & NO \\
     \#restricted-list & $\N$ & $\overline{\N}\times\mathbf{2}$  & $\#$restricted-list-$H$-COLORING \cite{diaz2005restrictive} & NO \\ \hline
\end{tabular}
\caption{Examples of $H$-COLORING problems through semirings (with the notation $\mathbf{2}=\{0,1\}$).}
\label{tab:semiring generalisations}
\end{table}

We remark that the approach of generalizing the CSPs using semirings has already been studied by Bistarelli et al.\ \cite{bistarelli1997semiring}, introducing the structure of {\em $c$-semirings}, and the computational problem {\em SCSP}. However, Bistarelli et al.\ focused  on  generalisations of CSP  involving an optimisation process, requiring a relation $a \leq b$ stating that $a$ is ``preferable'' to $b$, achieved by defining  $a\leq b \iff a+b=b$ and requiring the $+$ operation to be idempotent.
Wilson~\cite{DBLP:conf/ijcai/Wilson05} also defined an equivalent framework that generalizes CSP with semirings. Our contribution with respect to these alternative frameworks is the introduction of the notion of a pre-morphism and its clear link to the operations $\uplus$ and $\Join$, which are heavily used in the algorithms in Section~\ref{sec:tractability ctww}.

\section{Complexity of $\Omega$($H$-($\R$-MORPHISM)) via Component Twin-Width}\label{sec:tractability ctww}

In this section we analyze the complexity of $\Omega$($H$-($\R$-MORPHISM)), and thus of $\Omega$(BINARY-CSP($\Gamma$)), with respect to component twin-width. In Subsection~\ref{sec:fpt} we consider input graphs with bounded component twin-width, and  in Subsection~\ref{sec:fine_grained} we consider target templates $H$ with bounded component-twin width. In Subsection~\ref{sec:consequences} illustrate our tractability results on examples from the literature.

\subsection{Parameterized Complexity and Fixed-Parameter Tractability}
\label{sec:fpt}
We begin by proposing a dynamic programming algorithm  applicable to BINARY-CSP and its generalizations in the semiring framework.
To simplify the statement of Theorem~\ref{thm:Solving Omega H-(R-MORPHISM) parameterized}, say that a pre-morphism $\Omega$ is {\em  corestriction independent} if for all $V_1$ and $V_2$, $W\in B^{V_1\times V_2}$, and for all subsets $S\subseteq V_1$ and $T\subseteq V_2$, $f\in T^S$, and $T'\subseteq T$ with $f(S)\subseteq T'$, the corestriction $f|^{T'}\in (T')^S$ satisfies $\Omega_W(\{f\})=\Omega_W(\{f|^{T'}\})$. This rather weak assumption is satisfied by every strong pre-morphism we considered: essentially, it only requires that the singleton values of a  function $f$, do not depend on the vertices of the target graph that are not in the image of $f$.

\begin{restatable}{theorem}{parametrized}
\label{thm:Solving Omega H-(R-MORPHISM) parameterized}

Let $(A,+,\times,0_A,1_A)$ be a semiring, $\Omega$ a poly-time computable strong $A$-pre-morphism with weights in a set $B$, $H$ an $\red$-free edge-labelled graph, $X$ a finite set,  and $\R\subseteq X\times X_H$. Assume that $\Omega$ is corestriction independent. Then, for every instance $G$ on $n\geq 1$ vertices, Algorithm \ref{algo:main algorithm parameterized} with fixed, $H$, $\R$ and $\Omega$ solves $\Omega$($H$-($\R$-MORPHISM)) in $O(((2^{|V_H|}-1)^{ctww(G)})\times|V_G|^2)$ time, provided that an optimal contraction sequence $\mathbb{G}$ of $G$ is given.

\end{restatable}

Algorithm~\ref{algo:main algorithm parameterized} can be found in Appendix~\ref{sec:main algorithm parameterized}, and uses dynamic programming along with optimal contraction sequences to achieve the desired solution(s). We are able to deal with semiring generalisations since, for an $\red$-connected component $C$  in the contraction sequence, instead of just keeping in memory whether a function $\gamma: C\mapsto \P(V_H)\setminus\emptyset$ is a profile (as it is done  in~\cite{bonnet2022twin6}), we store the value by $\Omega$ of the set of partial solutions that induce this profile (i.e., the set of $\R$-morphisms $f:C\mapsto V_H$ such that for all $S\in C, f(S) = \gamma(S)$). We see that the algorithm by Bonnet et al.~\cite{bonnet2022twin6} solving the $q$-COLORING problem in FPT time parameterized by component twin-width is the particular case where the pre-morphism considered is over the semirings of Booleans, and maps $\emptyset$ to $0$ and any other set to $1$ (see Lemma \ref{lemma:decision dioid}). Also, we are able to deal with arbitrary edge-labelled graphs $H$ instead of only the $q$-clique $K_q$ by replacing the update of sets of profiles by Bonnet et al. \cite{bonnet2022twin6} that checks the absence of a black edge, by the more general test of {\em feasibility} (see Appendix \ref{app:feasibility}). Thus, Algorithm~\ref{algo:main algorithm parameterized} is much more general than the aforementioned algorithm by Bonnet et al.\ since it applies to arbitrary binary constraints rather than the specific template $K_q$, and is applicable to generalized problems described by the pre-morphism $\Omega$ (counting, with weights, and so on).

It may also be interesting to remark note that the complexity of Algorithm~\ref{algo:main algorithm parameterized} depends neither on the semiring $A$, nor on the pre-morphism $\Omega$, i.e., the generalized problems do not impact the running time. 

\subsection{Upper Bounds on Fine-Grained Complexity}
\label{sec:fine_grained}

In Section~\ref{sec:fpt} we exploited the contraction sequence of an instance $\mathcal{I}$ to obtain an FPT algorithm with respect to the component twin-width of $G(\mathcal{I})$. We now turn to the dual question of constructing an improved (exponential time) algorithm with respect to the component twin-width of $H(\Gamma)$.
Note that the computation of an optimal contraction sequence of $H(\Gamma)$ can be seen as a form of pre-computation since it is independent of the instance. Also, when working with $H$ instead of $G$, our algorithm will have to guess preimages of subsets $T$ of $V_H$ instead of images of subsets $S$ of $V_G$. Since preimages of pairwise disjoint subsets are  pairwise disjointed, our algorithm also applies to semiring pre-morphisms that are not strong.

\begin{restatable}{theorem}{finegrained}\label{thm:Solving Omega H-(R-MORPHISM)}

Let $(A,+,\times,0_A,1_A)$ be a semiring, and $\Omega$ a $A$-pre-morphism with weights in a set $B$, $H$ an $\red$-free edge-labelled graph, $\mathbb{H}$ an optimal contraction sequence of $H$, and $X$ a finite set, $\R\subseteq X\times X_H$. Then, Algorithm \ref{algo:main algorithm} with fixed $\mathbb{H}$, $\R$ and $\Omega$ solves $\Omega$($H$-($\R$-MORPHISM)) in time $O( (ctww(H)+2)^{|V_G|}\times |V_G|^2)$ every intance $G$ on $n\geq 1$ vertices.

\end{restatable}

Algorithm~\ref{algo:main algorithm}  has similarities with the algorithm given by Wahlstr\"om \cite{wahlstrom2011new} which solves $\#H$-COLORING in time $O^*((2cw(H)+1)^n)$ on input graphs $G$ on $n$ vertices (with $cw(H)$ the clique-width of $H$). However, our algorithm has two significant advantages. First, it is applicable to arbitrary BINARY-CSP problems, and its generalized problems in the semiring framework (without impacting the running time). Second, even though clique-width and component twin-width are functionally equivalent, even a minor increase of the width parameter can significantly change the run time of the algorithm. For example, the problem of counting the number of homomorphisms into a cycle of length $k$, \#$C_k$-COLORING, can be solved in $O^*(5^n)$ time by Algorithm~\ref{algo:main algorithm} but requires $O^*(6^n)$ time by the clique-width algorithm. Additional examples are provided in Section~\ref{sec:consequences}.

\subsection{Consequences} 
\label{sec:consequences}

The tractability of many semiring generalisations of $H$-COLORING and BINARY-CSP$(\Gamma)$ easily follow from Theorem~\ref{thm:Solving Omega H-(R-MORPHISM) parameterized}.

\begin{corollary}

Let $\Gamma$ be a set of binary relations over a finite domain. Then, BINARY-CSP($\Gamma$), $\#$BINARY-CSP($\Gamma$), $\#$list-BINARY-CSP($\Gamma$), $\#$cost-list-BINARY-CSP($\Gamma$) with weights are FPT parameterised by the component twin-width of the instance.

\end{corollary}

This strongly generalizes many results in the literature. For example, Kobler and Rotics \cite{kobler2003edge} proved that costs-list-$q$-COLORING is FPT when parameterized by clique-width (recall that clique-width and component twin-width are functionally equivalent on graphs~\cite{bonnet2022twin6}). Similarly, we strongly generalize Wahlstr\"om's FPT algorithm for \#$H$-COLORING (with respect to clique-width) since we can handle arbitrary binary constraints as well as the extended problems. This also supplements to the well known result that, for every graph $H$, $H$-COLORING is FPT when parameterized by clique-width, by solving also counting versions. This is a corollary derived from Courcelle et al. \cite{courcelle1990monadic} algorithm that solves an optimization version of the problem of checking whether a valuation over the vertices of structure is a model of a fixed monadic second-order logic formula in FPT time (parameterized by clique-width), see Table \ref{tab:tractability results}.

\begin{table}
    \centering
    \begin{tabular}{c c c}
    \textbf{Generalisation} & Target for Courcelle \cite{DBLP:journals/mst/CourcelleMR00} & Target for Theorem \ref{thm:Solving Omega H-(R-MORPHISM) parameterized} \\
    \hline
        $H$-COLORING & YES & YES \\
        $\#H$-COLORING &  & YES \\
        list-$H$-COLORING & YES & YES \\
        \#list-$H$-COLORING &  & YES \\
        cost-$H$-COLORING & YES & YES \\
        \#cost-$H$-COLORING &  & YES \\
        \#list-cost-$H$-COLORING &  & YES \\
        weighted-$H$-COLORING & YES & see discussion in Section \ref{sec:conclusions} \\
        \#weighted-$H$-COLORING &  & see discussion in Section \ref{sec:conclusions} \\
        \#list-weighted-$H$-COLORING &  & see discussion in Section \ref{sec:conclusions} \\
        restricted-$H$-COLORING & YES & see discussion in Section \ref{sec:conclusions} \\
        \#restricted-$H$-COLORING &  & see discussion in Section \ref{sec:conclusions} \\
        \#list-restricted-$H$-COLORING &  & see discussion in Section \ref{sec:conclusions} \\
        \#list-restricted-cost-$H$-COLORING &  & see discussion in Section \ref{sec:conclusions} \\
        \#list-restricted-weighted-$H$-COLORING &  & see discussion in Section \ref{sec:conclusions} \\
        \hline
    \end{tabular}
    \caption{Tractability results parameterized by the (functionally equivalent) parameters clique-width/component twin-width.}
    \label{tab:tractability results}
\end{table}

We can also use Theorem \ref{thm:Solving Omega H-(R-MORPHISM)} to derive upper bounds on the complexities of several generalisations of $H$-COLORING problems (for some specific values of $H$) through semirings that improve previously know results. These results are summarized in Table \ref{tab:upper bounds} and are straightforward consequences of Theorem \ref{thm:Solving Omega H-(R-MORPHISM) parameterized} and Theorem \ref{thm:Solving Omega H-(R-MORPHISM)} (but explicitly demonstrated in in Appendix \ref{app:examples of pre-morphisms}).

\section{Conclusions and Perspectives}
\label{sec:conclusions}

We investigated the complexity of binary constraint satisfaction problems under the lens of the component twin-width parameter. In order to obtain as general results as possible, we considered several frequently occurring problem extensions, e.g., counting, allowing weights, cost and list constraints, which we formulated in a unifying semiring framework which greatly simplified the algorithmic results. Importantly, we obtained two novel algorithms by bounding either the class of input instances, or the constraint template, and presented several instances where our approach beats both the best known upper bound (e.g., counting homomorphisms to cycles) as well as improving upon earlier algorithms making use of tree-width and clique-width.
These results raise several questions for future work:

{\flushleft\bf Generalized problems.} Even though Theorem \ref{thm:Solving Omega H-(R-MORPHISM) parameterized} and Theorem \ref{thm:Solving Omega H-(R-MORPHISM)} are very general algorithms applicable to broad classes of binary constraints it is still tempting to generalize them to even wider classes of problems. For instance, even though Theorem \ref{thm:Solving Omega H-(R-MORPHISM) parameterized} does not apply to every generalisation of BINARY-CSP presented in Table \ref{tab:tractability results} (the one involving $weighted$ and $restricted$), it is still possible to prove that these problems are FPT parametrized by component twin-width. It is sufficient to modify Algorithm \ref{algo:main algorithm parameterized} by adding to the tabular $\overline{\text{OMEGA}}$ $m$ entries corresponding to the $m$ vertices of $H$, corresponding to the weights/cardinal (when considering the weighted/restricted generalisations) of preimages reached by every vertex of $H$, which allows dynamic programming. We expect that implementing an algebraic structure over the set of weights $B$ would make possible to reformulate Theorem \ref{thm:Solving Omega H-(R-MORPHISM) parameterized} and Algorithm \ref{algo:main algorithm parameterized} in order to include these kind of algorithms and tractability results as well. To go even further, it would be interesting to generalize Theorem \ref{thm:Solving Omega H-(R-MORPHISM) parameterized} and Theorem \ref{thm:Solving Omega H-(R-MORPHISM)} to constraints of arbitrary arity. One possibility is to express such CSPs by generalising the concept of edge-labelled graphs to ``edge-labelled hypergraphs'', labelling arbitrary large tuples over $V_G$. Are the usual parameters on hypergraphs, and in particular generalizations of (component) twin-width, applicable to edge-labelled hypergraphs? Variants of CSP could also be considered. For instance, can we solve the PROMISE-BINARY-CSP problems (see Barto et al. \cite{barto2021algebraic} for details) similarly? A promise constraint satisfaction problem (PCSPs) requires two finite similar structures $\mathcal{A}$ and $\mathcal{B}$ with an homomorphism $h:\mathcal{A}\rightarrow \mathcal{B}$, and asks, given a structure $\mathcal{I}$ whether $\mathcal{I}\rightarrow \mathcal{A}$ or $\mathcal{I}\cancel{\rightarrow} \mathcal{B}$, with the promise that these two statements are not both false (we know that they are not both true from the existence of $h$). The algebraic approach proposed by Barto et al. \cite{barto2021algebraic} already led to interesting hardness results, such as the NP-hardness of the distinguishment of the $q$-colourable graphs from those that are not $(2q- 1)$-colourable. As a complete dichotomy theorem still eludes us, it would be interesting to check whether techniques from parameterized and fine-grained complexity can be applied in the promise setting.

{\flushleft\bf Comparison with clique-width and other parameters.} Other graph parameters can probably be efficiently extended to edge-labelled graphs and thus to BINARY-CSP, while preserving the soundness of the algorithms working on this parameters. We believe that clique-width should be a good candidate since it is functionally equivalent on graphs~\cite{bonnet2022twin6}, with the main structural subtlety  being the creation of labelled edges in $k$-expressions. Assuming that we have been able to extend clique-width to edge-labelled graphs, does it stay functionally equivalent to component twin-width on edge-labelled graphs? We also believe that investigating the parameterized complexity of semirings generalisations of (BINARY-)CSP with other parameters than component twin-width could lead to similar interesting results. However, we have been unable to produce any meaningful results with twin-width rather than component twin-width. Can this difficulty be formalized into a concrete lower bound?
{\flushleft\bf Algebraic developments.} The algebraic generalisations of BINARY-CSP through semiring pre-morphism are inconvenient in certain aspects. For instance, as discussed earlier, it  would be interesting to describe convenient structures over the sets of weights. Most importantly, we still lack algebraic operations that combine semiring pre-morphisms together, in order to automatically handle combinations of semiring generalisations without redefining a new semiring pre-morphism each time. For example, it would be desirable to be able to build the $\N$-pre-morphism ``$\#_{\text{list}}$'' (which leads to the counting-list generalisation, see Lemma \ref{lemma:list counting dioid}) using the $\N$-pre-morphism ``$\#$'' (which leads to the counting generalisation, see Lemma \ref{lemma:counting dioid}) and the $\mathbf{2}$-pre-morphism ``$\text{list}$'' (which leads to the list generalisation, see Lemma \ref{lemma:list dioid}). Moreover, we noticed that every semiring used in this paper is even a dioid. Can we take advantage of the additional properties of dioids? More generally, can we extend BINARY-CSP via other algebraic structures?

\newpage

\appendix

\section{Proofs of Section \ref{sec:hom_equivalence}}\label{sec:proof of hom_equivalence}

We give here a proof of Theorem \ref{thm:Binary-CSP to R-morphism}. It follows naturally from the constructions $H(\Gamma)$, $\R_{\Gamma}$ and $G(\mathcal{I})$.

\BinaryCSPtoRmorphism*

\begin{proof}\label{proof:Binary-CSP to R-morphism}

$f$ is a solution of the instance $\mathcal{I}$ of BINARY-CSP$(\Gamma)$ $\underset{\text{By definition of BINARY-CSP}(\Gamma)}{\iff}$

For all $(u,v)\in V^2$, for all constraint $R(u,v)$ of $\mathcal{I}$, $(f(u),f(v))\in R$
$\underset{\text{By definition of }l_{\mathcal{I}(u,v)}}{\iff}$

For all $(u,v)\in V^2$, for all $R\in l_{\mathcal{I}}(u,v)$, $(f(u),f(v))\in R$
$\underset{\text{By definition of }l_{\Gamma}(f(u),f(v))}{\iff}$

For all $(u,v)\in V^2$, for all $R\in l_{\mathcal{I}}(u,v)$, $R\in l_{\Gamma}(f(u),f(v))$
$\underset{\text{By definition of inclusion}}{\iff}$

For all $(u,v)\in V^2$, $l_{\mathcal{I}}(u,v) \subseteq l_{\Gamma}(f(u),f(v))$
$\underset{\text{By definition of }\R_{\Gamma}}{\iff}$

For all $(u,v)\in V^2$, $(l_{\mathcal{I}}(u,v),l_{\Gamma}(f(u),f(v))) \in \R_{\Gamma}$
$\underset{\text{By definition of }H(\Gamma)\text{-(}\R_{\Gamma}\text{-MORPHISM)}}{\iff}$

$f$ is a solution of the instance $G(\mathcal{I})$ of $H(\Gamma)$-($\R_{\Gamma}$-MORPHISM)

\end{proof}

Similarly, we give a proof to Theorem \ref{thm:R-morphism to Binary-CSP} which essentially follows from the definitions  of $\Gamma_{H,\R}$ and $\mathcal{I}(G)$.

\RmorphismtoBinaryCSP*

\begin{proof}\label{proof:R-morphism to Binary-CSP}

$f$ is a solution of the instance $\mathcal{I}(G)$ of BINARY-CSP($\Gamma_{H,\R}$)
$\underset{\text{By definitions of BINARY-CSP}(\Gamma_{H,\R})\text{ and }\mathcal{I}(G)}{\iff}$

For all $(u,v)\in (V_G)^2$, $(f(u),f(v))\in R_{l_G(u,v)}$
$\underset{\text{By definition of }R_{l_G(u,v)}}{\iff}$

For all $(u,v)\in (V_G)^2$, $(l_G(u,v),l_H((f(u),f(v))))\in \R$
$\underset{\text{By definition of }H\text{-}(\R\text{-MORPHISM})}{\iff}$

$f$ is a solution of the instance $G$ of $H$-($\R$-MORPHISM).

\end{proof}

\section{Feasibility}\label{app:feasibility}

Looking at Definition \ref{def:Meaning Contraction}, we can interpret the $\red$-edges of $H'$ as a loss of information. Since we want to study the $\R$-morphisms of every subgraph of $G$ to some wisely chosen subgraphs of $H$, it seems natural to choose the subgraphs of $H$ of the form $H[T_1\uplus\dots\uplus T_p]$, when $\{T_1,\dots, T_p\}$ is a $\red$-connected component of a graph $H_k$ ($k\in [m]$) of the contraction. This is the main reason why component twin-width affects the complexity of the algorithm.

\begin{definition}\label{def:feasibility}

Let $G$ and $H$ be two $\red$-free edge-labelled graphs, $\R\subseteq X_G \times X_H$, $H'$ a contraction of $H$, $\T=(T_1,\dots,T_p)$ a tuple of $p$ different vertices of $H'$, and $\S=(S_1,\dots,S_p)$ a tuple of $p$ pairwise disjointed subsets of $V_G$ (possibly empty).
We say that $\S$ is $\R$-{\em feasible} with respect to $\T$ and we denote $\S\preceq_{\R}\T$ if for all $(i,i')\in [p]^2$, we have:
\[l_{H'}(T_i,T_{i'})\neq \red \implies \forall (u,v)\in S_i\times S_{i'}, (l_G(u,v),l_{H'}(T_i,T_{i'})) \in \R.\]

\end{definition}

We also define a symmetrical notion when a contraction of $G$ is considered instead.

\begin{definition}\label{def:making feasibility}

Let $G$ and $H$ be two $\red$-free edge-labelled graphs, $\R\subseteq X_G \times X_H$, $G'$ a contraction of $G$, $\S=(S_1,\dots,S_p)$ a tuple of $p$ different vertices of $G'$, and $\T=(T_1,\dots,T_p)$ a tuple of $p$ non-empty subsets of $V_H$ (not necessarily pairwise disjointed).
We say that {\em $\T$ makes $\S$ $\R$-feasible}, and we denote $\T \succeq_{\R} \S$ if for all $(i,i')\in [p]^2$, we have:

\[l_{G'}(S_i,S_{i'})\neq \red \implies \forall (a,b)\in T_i\times T_{i'}, (l_G(S_i,S_{i'}),l_H(a,b)) \in \R.\]
    
\end{definition}

In the Definition \ref{def:feasibility}, we ask for the elements of $\S$ to be pairwise disjointed, whereas in Definition \ref{def:making feasibility}, we require the elements of $\T$ to be non-empty. The reason is that when employing Definition \ref{def:feasibility}, the elements of $\S$ will play the role of the preimages of the elements of $\T$, and will therefore be pairwise disjointed (since the elements of $\T$ are non-empty and pairwise disjointed), whereas in Definition \ref{def:making feasibility}, the elements of $\T$ will play the role of the images of the elements of $\S$, and will therefore be non-empty (since the elements of $\S$ are non-empty and pairwise disjointed).

\begin{remark}

Using Definition \ref{def:Meaning Contraction}, we notice the truthness of the proposition ``$\S\preceq_{\R}\T$'' and ``$\T\succeq_{\R}\S$'' does not depend on the contraction $H'$ or $G'$.

\end{remark}

\section{Soundness of Algorithm \ref{algo:main algorithm}}

\begin{algorithm}
\caption{Solving $\Omega$($H$-($\R$-MORPHISM)) ie. $\Omega$(BINARY-CSP($\Gamma$)) (fine-grained version)}
\label{algo:main algorithm}
Create a tabular OMEGA filled with $0_A$\

\For{$S\subseteq V_G$, $a\in V_H$}
{
\If{$\forall (u,v)\in S^2, (l_G(u,v),l_H(a,a))\in\R$}
{OMEGA$\begin{bmatrix} S & \{a\} \end{bmatrix}\gets \Omega_W(\{f^S_{\{a\}}\})$ (with $f^S_{\{a\}} : \begin{matrix} S & \mapsto & \{a\} \\ u &\mapsto & a \end{matrix}$)}
}

\For{$k=m-1$ downto $1$}
{
$(T_p,T_{p+1})\leftarrow$ contracted pair in the contraction $H_{k+1} \rightarrow H_k$ of $\mathbb{H}$

$T_0 \leftarrow$ contraction of $T_p$ and $T_{p+1}$ in $H_k$

$C=\{T_0,T_1,\ldots T_{p-1}\}$ the $\red$-connected components of $H_k$ containing $T_0$

$C_1 \uplus \ldots\uplus C_q = \{T_1,\ldots,T_{p-1},T_p,T_{p+1}\} $ be the partitionning of $(C\setminus \{T_0\})\cup \{T_p,T_{p+1}\}$ into $\red$-connected components in $H_{k+1}$

\For{$j=1$ to $q$}
{
Define $I_j\subseteq [p]$ such that $C_j=\{T_i, i\in I_j \}$

$I_j=:\{ I_j[1],\dots,I_j[p_j] \}$ with $p_j=|I_j|$
}

\For{$S_0,S_1,\ldots,S_{p-1}\subseteq V_G$ pairwise disjointed}
{
\For{$S_p\uplus S_{p+1}$ partitionning $S_0$}
{
\If{$(S_1,\dots,S_{p-1},S_p,S_{p+1}) \preceq_{\R}(T_1,\dots,T_{p-1},T_p,T_{p+1})$}
{OMEGA$\begin{bmatrix} S_0 & T_0 \\  S_1 & T_1 \\ \vdots & \vdots  \\ S_{p-1} & T_{p-1}\end{bmatrix}+=$ OMEGA$\begin{bmatrix} S_{I_1[1]} & T_{I_1[1]} \\  \vdots & \vdots  \\ S_{I_1[p_1]} & T_{I_1[p_1]}\end{bmatrix}\times\dots\times$ OMEGA$\begin{bmatrix} S_{I_q[1]} & T_{I_q[1]} \\  \vdots & \vdots  \\ S_{I_q[p_q]} & T_{I_q[p_q]}\end{bmatrix}$}
}
}
}
\Return{{\em OMEGA}$\begin{bmatrix} V_G & V_H \end{bmatrix}$}
\end{algorithm}

\begin{definition}

Let $G$ and $H$ two $\red$-free edge-labelled graphs on $n$ and $m$ vertices respectively. Let $\R\subseteq X_G \times X_H$.
Let $\T=(T_1,\dots,T_p)$ be a tuple of $p$ pairwise disjointed subsets of $V_H$ (in particular, $T_1,\dots,T_p$ can be different vertices of a contraction $H'$ of $H$) and $T=T_1\uplus\dots\uplus T_p$.
Let $\S=(S_1,\dots,S_p)$ be a tuple of $p$ pairwise disjointed subsets of $V_G$, and $S=\cup \S=S_1\uplus\dots\uplus S_p$. We denote by $\R^{\S}_{\T}$ the set

\[\R^{\S}_{\T}:=\{ f:G[S]\underset{\R}{\rightarrow} H[T] \mid f(S_1)\subseteq T_1,\dots,f(S_p)\subseteq T_p \}\]

\end{definition}

\begin{lemma}\label{lemma:empty set,join lemma}

Let $G$ and $H$ two $\red$-free edge-labelled graphs, on $n$ and $m$ vertices respectively. Let $H'$ be a contraction of $H$, $\R\subseteq X_G\times X_H$, $C_1,\dots,C_q$ be $q$ different $\red$-connected components of $H'$ ($q\geq 1$), $\{T_1,\dots,T_p\} = C_1\uplus\dots\uplus C_q$ and let $\T=(T_1,\dots,T_p)$, $\S=(S_1,\dots,S_p)$ be a tuple of $p$ pairwise disjointed subsets of $V_G$. Let for all $j\in q$, $I_j\subseteq [p]$ be such that $C_j=\{T_i,i\in I_j\}$. Then 

\[\R^{\S}_{\T} = \left\{
    \begin{array}{ll}
        \emptyset & \text{ if } \S\cancel{\preceq}_{\R}\T
        \\
        \R^{\S_{I_1}}_{\T_{I_1}}\Join \dots\Join \R^{\S_{I_q}}_{\T_{I_q}} & \text{ if } \S\preceq_{\R}\T
    \end{array}\right.\]

\end{lemma}

\begin{proof}

First, assume that $\S\cancel{\preceq}_{\R}\T$:
Assume by contradiction that there exists $f\in \R^{\S}_{\T}$, i.e.\ $f$ is an $\R$-morphism from the edge-labeled graph $G[S_1\uplus\dots\uplus S_p]$ to the edge-labeled graph $H$ and $\forall i\in [p], f(S_i)\subseteq T_i$. Since $\S$ is not $\R$-feasible with respect to $\T$, there exists a pair $(i,i')\in [p]^2$ such that:

\begin{itemize}
    \item $l_{H'}(T_i,T_{i'})\neq \red$,
    \item there exists $u\in S_i$ and $v\in S_{i'}$ with $(l_G(u,v),l_{H'}(T_i,T_{i'}))\notin \R$.
\end{itemize}

By Definition \ref{def:Meaning Contraction}, $\forall (a,b)\in T_i\times T_{i'}, l_H(a,b)= l_{H'}(T_i,T_{i'})$. Notice that, by definition of $f\in\R^{\S}_{\T}$, $(f(\underbrace{u}_{\in S_i}),f(\underbrace{v}_{\in S_{i'}})) \in T_i \times T_{i'}$. Thus, $l_H(f(u),f(v))=l_{H'}(T_i,T_{i'})$. Since $f$ is an $\R$-momorphism, we have $(l_G(u,v),\underbrace{l_H( f(u),f(v) )}_{=l_{H'}(T_i,T_{i'})} ) \in \R$, which contradicts the definition of $(u,v)$. We have a contradiction, which proves that $\R^{\S}_{\T} = \emptyset$.

Second, assume that $\S\preceq_{\R}\T$:

Take $f\in \R^{\S}_{\T}$ and $j\in [q]$. Since restrictions of $\R$-morphisms are $\R$-morphisms, $f|_{\cup\S_{I_j}}$ is an $\R$-morphism. Moreover, since $f\in \R^{\S}_{\T}$, for all $j\in I_j$, $f(S_j)\subseteq T_j$. We deduce that $f|_{\cup\S_{I_j}}\in \R^{\S_{I_j}}_{\T_{I_j}}$. This proves that $\R^{\S}_{\T} \subseteq \R^{\S_{I_1}}_{\T_{I_1}}\Join \dots\Join \R^{\S_{I_q}}_{\T_{I_q}}$.

We now prove the reverse. Let $(f_1,\dots,f_q)\in \R^{\S_{I_1}}_{\T_{I_1}}\times \dots\times \R^{\S_{I_q}}_{\T_{I_q}}$, and let $f=f_1\Join\dots\Join f_q$. We will prove that $f\in \R^{\S}_{\T}$. Clearly, by definition of the $\R^{\S_{I_j}}_{\T_{I_j}}$ for $j\in [q]$, we have, for all $i\in [p], f(S_i)\subseteq T_i$ (knowing that $(I_1,\dots,I_q)$ is a partition of $[p]$). There only remains to prove that $f$ is an $\R$-morphism.
Let $S=S_1\uplus\dots\uplus S_p$ and take $(u,v)\in S^2$. We will prove that $(l_G(u,v),l_H(f(u),f(v)))\in \R$. Let $(i,i')\in [p]^2$ be such that $u\in S_i$ and $v\in S_{i'}$.

\begin{enumerate}

\item If there exists $j\in [q]$ such that $(i,i')\in (I_j)^2$ (ie. if $T_i$ and $T_{i'}$ belong to the same $\red$-connected component $C_j$), then, $(u,v)\in (S_{I_j})^2$. It follows by definition of $f$ that $(f(u),f(v))=(f_j(u),f_j(v))$, and then $(l_G(u,v),l_H(f(u),f(v)))\in\R$ because $f_j$ is an $\R$-morphism.

\item Else, by definitions of $I_j$ and $C_j$ for $j\in [q]$, $T_i$ and $T_{i'}$ are not $\red$-connected in $H'$. We deduce that, in particular, $l_{H'}(T_i,T_{i'})\neq \red$. Using Definition \ref{def:Meaning Contraction}: $\forall (a,b)\in T_i\times T_{i'}, l_H(a,b)=l_{H'}(T_i,T_{i'})$. Then, $(f(u),f(v))=(f_i(u),f_{i'}(v))\in T_i\times T_{i'}$, thus $l_H(f(u),f(v))=l_{H'}(T_i,T_{i'})$. Using $\S\preceq_{\R}\T$, we have $(l_G(u,v),\underbrace{l_{H'}(T_i,T_{i'})}_{=l_H(f(u),f(v))})\in \R$.
\end{enumerate}

We have proven that $(l_G(u,v),l_H(f(u),f(v)))\in\R$. Hence, $f$ is indeed an $\R$-morphism, which concludes the proof.

\end{proof}

\begin{lemma}\label{lemma:split}

Let $G$ and $H$ be two $\red$-free edge-labelled graphs, and let $\R\subseteq X_G\times X_H$, $T_0,T_1,\dots,T_{p-1}$ be $p$ pairwise disjointed subsets of $V_H$, and let $\T=(T_1,\dots,T_{p-1})$, $S_0,S_1,\dots,S_{p-1}$ be $p$ pairwise disjointed subsets of $V_G$, and let $\S=(S_1,\dots,S_{p-1})$, $T_p\uplus T_{p+1}=T_0$ be a partition of $T_0$.Then

\[\R^{\S,S_0}_{\T,T_0}
    =
\underset{\begin{matrix} S_p\uplus S_{p+1} = S_0 \\ (\S,S_p,S_{p+1}) \preceq_{\R} (\T,T_p,T_{p+1}) \end{matrix}}{\uplus}
\R^{\S,S_p,S_{p+1}}_{\T,T_p,T_{p+1}}
\]

\end{lemma}

\begin{proof}

Partitioning the set
$\R^{\S,S_0}_{\T,T_0}$ in equivalence classes with respect to the equivalence relation $\sim$ defined by:

\hfill

\[ \forall (f,f')\in
(\R^{\S,S_0}_{\T,T_0})^2,
f\sim f' \iff \underbrace{(f^{-1}(T_p),f^{-1}(T_{p+1}))}_{\text{ partition of } S_0} = \underbrace{(f'^{-1}(T_p),f'^{-1}(T_{p+1}))}_{\text{ partition of } S_0} \]

we obtain:

\[\R^{\S,S_0}_{\T,T_0}
    =
\underset{S_p\uplus S_{p+1} = S_0}{\uplus}
\R^{\S,S_p,S_{p+1}}_{\T,T_p,T_{p+1}}
\]

and then, the conclusion follows from Lemma \ref{lemma:empty set,join lemma}.

\end{proof}

\begin{theorem}\label{thm:main cardinal theorem}

Let $\Omega$ a semiring pre-morphism and $W$ a weight matrix for $\Omega$, $G$ and $H$ be two $\red$-free edge-labelled graphs, $(H_m,\dots,H_1)$ a contraction sequence of $H$ (with $m:=|V_H|$), and $k\in [m-1]$, $C_1,\dots,C_q$ be $q$ ($q\geq 1$) different $\red$-connected components of $H_{k+1}$ ($q\geq 1$). Let $\{T_1,\dots,T_{p-1},T_p,T_{p+1}\} = C_1\uplus\dots\uplus C_q$ and $\T=(T_1,\dots,T_{p-1})$, $S_0,S_1,\dots,S_{p-1}$ be $p$ pairwise disjointed subsets of $V_G$. Let $\S=(S_1,\dots,S_{p-1})$. Denote, for all $j\in q$, $I_j\subseteq [p+1]$ such that $C_j=\{T_i,i\in I_j\}$. Assume that $H_k$ is obtained from $H_{k+1}$ by contracting the two different vertices $T_p$ and $T_{p+1}$ to the vertex $T_0=T_p\uplus T_{p+1}$. Then, we have:

\[\Omega_W(\R^{\S,S_0}_{\T,T_0}) = \underset{\begin{matrix} S_p\uplus S_{p+1} = S_0 \\ (\S,S_p,S_{p+1})  \preceq_{\R}  (\T,T_p,T_{p+1}) \end{matrix}}{\sum}
    (\overset{q}{\underset{j=1}{\prod}} \Omega_W (\R^{(\S,S_p,S_{p+1})_{I_j}}_{(\T,T_p,T_{p+1})_{I_j}}))\]

(the notations $\sum$ and $\prod$ reffer to the sum and product of the semiring).

\end{theorem}

\begin{proof}

Using Lemma \ref{lemma:split}:

$\Omega_W(\R^{\S,S_0}_{\T,T_0})
    =
\Omega_W(\underset{\begin{matrix} S_p\uplus S_{p+1} = S_0 \\ (\S,S_p,S_{p+1})  \preceq_{\R}  (\T,T_p,T_{p+1}) \end{matrix}}{\uplus}
\R^{\S,S_p,S_{p+1}}_{\T,T_p,T_{p+1}}\,\,\,\,\,)
$

\hfill

using the first axiom of pre-morphisms:

\hfill

$\Omega_W(\R^{\S,S_0}_{\T,T_0})
    =
\underset{\begin{matrix} S_p\uplus S_{p+1} = S_0 \\ (\S,S_p,S_{p+1})  \preceq_{\R}  (\T,T_p,T_{p+1}) \end{matrix}}{\sum} \Omega_W(\R^{\S,S_p,S_{p+1}}_{\T,T_p,T_{p+1}})
$

\hfill

and we can apply iteratively the second axiom of pre-morphisms with $\F:=\R^{\S,S_p,S_{p+1}}_{\T,T_p,T_{p+1}}$ and $\forall j\in [q], \F_j:=\R^{(\S,S_p,S_{p+1})_{I_j}}_{(\T,T_p,T_{p+1})_{I_j}}$, thanks to Lemma \ref{lemma:empty set,join lemma} (since $(\S,S_p,S_{p+1})  \preceq_{\R}  (\T,T_p,T_{p+1})$) in order to write:

\hfill

$\Omega_W(\R^{\S,S_0}_{\T,T_0}) = \underset{\begin{matrix} S_p\uplus S_{p+1} = S_0 \\ (\S,S_p,S_{p+1})  \preceq_{\R}  (\T,T_p,T_{p+1}) \end{matrix}}{\sum}
    (\overset{q}{\underset{j=1}{\prod}} \Omega_W (\R^{(\S,S_p,S_{p+1})_{I_j}}_{(\T,T_p,T_{p+1})_{I_j}}))$

\hfill

Which was what we wanted to prove.

\end{proof}

\begin{lemma}\label{lemma:loop invariant}

Let $(A,+,\times,0_A,1_A)$ be a semiring, $B$ a set, $\Omega$ a $A$-pre-morphism with weights in $B$, $W\in B^{V_G\times V_H}$, $G$ and $H$ be two $\red$-free edge labelled graphs, and $(H_m,\dots,H_1)$ be a contraction sequence of $H$.
    
For all $k\in [m]$, for all $\red$-connected component $\{T_1,\dots,T_p\}$ of $H_k$, and for all $S_1,\dots,S_p$ pairwise disjointed subsets of $V_G$, after the iteration of index $k$ (and before the iteration of index $k-1$) of the loop ``for $k=m-1$ downto 1'' of Algorithm \ref{algo:main algorithm}, OMEGA$\begin{bmatrix} S_1 & T_1 \\  \vdots & \vdots  \\ S_p & T_p\end{bmatrix}$ contains $\Omega_W(\R^{S_1,\ldots,S_p}_{T_1,\ldots,T_p})$.

\end{lemma}

\begin{proof}

We proceed by induction over $k=m$ down to $1$.

\underline{Initialisation:} For $k=m$, it comes the loop: ``for $S\subseteq V_G$'', noticing that the $\red$-connected components of $H_m$ are exactly its singletons of vertices, ie, the $\{a\}$ for $a\in V_H$, and noticing that, for all $S\subseteq V_G$, $\R^S_{\{a\}} = \{f^S_{\{a\}}\}$ if $\forall (u,v)\in S^2, (l_G(u,v),l_H(a,a))\in \R$, and $\R^S_{\{a\}}=\emptyset$ otherwise, recalling that $\Omega_W(\emptyset)=0_A$ (third axiom of pre-morphisms).

\underline{Hereditary:} Assume the lemma is true for $k+1$. Assume that the vertices merged in the contraction $H_{k+1}\rightarrow H_k$ are called $T_p$ and $T_{p+1}$ and that the merged vertex is called $T_0$. Notice that the only $\red$-connected component of $H_k$ that is not also a $\red$ connected component of $H_{k+1}$ is the one that contains $T_0$. Call it $C=\{T_0,T_1,\dots,T_{p-1}\}$. Let $C_1 \uplus \ldots\uplus C_q := \{T_1,\ldots,T_{p-1},T_p,T_{p+1}\} $ be the partitioning of $(C\setminus \{T_0\})\cup \{T_p,T_{p+1}\}$ into $\red$-connected component in $H_{k+1}$. For every $j\in [q]$, recalling that $C_j=\{ T_{I_j[1]},\dots,T_{I_j[p_j]}\}$ (by definition of $I_j$, and denoting $p_j:=|I_j|$), we have by the induction hypothesis: OMEGA$\begin{bmatrix} S_{I_j[1]} & T_{I_j[1]} \\  \vdots & \vdots  \\ S_{I_j[p_j]} & T_{I_j[p_j]}\end{bmatrix}$ contains $\Omega_W(\R^{\S_{I_j}}_{\T_{I_j}})$, where, we recall, $\R^{\S_{I_j}}_{\T_{I_j}}$ is a notation for $\R^{S_{I_j[1]},\dots,S_{I_j[p_j]}}_{T_{I_j[1]},\dots,T_{I_j[p_j]}}$. Then, we see that, at the end of the loop of index $k$ (and thus at the beginning of the loop of index $k-1$), ``for all $S_0,\dots,S_{p-1}$ subsets of $V_G$ pairwise disjointed'': OMEGA$\begin{bmatrix} S_0 & T_0 \\  S_1 & T_1 \\ \vdots & \vdots  \\ S_{p-1} & T_{p-1}\end{bmatrix}$ contains $\underset{\begin{matrix} S_p\uplus S_{p+1} = S_0 \\ (\S,S_p,S_{p+1})  \preceq_{\R}  (\T,T_p,T_{p+1}) \end{matrix}}{\sum}
    (\overset{q}{\underset{j=1}{\prod}} \Omega_W (\R^{(\S,S_p,S_{p+1})_{I_j}}_{(\T,T_p,T_{p+1})_{I_j}}))$ which, by Theorem \ref{thm:main cardinal theorem}, equals $\Omega_W(\R^{\S,S_0}_{\T,T_0})$. OMEGA$\begin{bmatrix} S_0 & T_0 \\  S_1 & T_1 \\ \vdots & \vdots  \\ S_{p-1} & T_{p-1}\end{bmatrix}$ contains $\Omega_W(\R^{\S,S_0}_{\T,T_0})$ at the end of the loop of index $k$. Recall that $\R^{\S,S_0}_{\T,T_0}$ is a notation for $\R^{S_1,\dots,S_{p-1},S_0}_{T_1,\dots,T_{p-1},T_0} = \R^{S_0,S_1,\dots,S_{p-1}}_{T_0,T_1,\dots,T_{p-1}}$. This concludes the proof.

\end{proof}

\finegrained*

\begin{proof}

The soundness of Algorithm \ref{algo:main algorithm} is a consequence of Lemma \ref{lemma:loop invariant}, noticing that the set of solutions of the instance $G$ of $H$-($\R$-MORPHISM) is exactly $\R^{V_G}_{V_H}$, and that $\{V_H\}$ is a connected component of $H_1$.

For the complexity, note that there are $(p+2)^{|V_G|}$ ways to choose $S_0,S_1,\dots,S_{p-1}$ and $S_p,S_{p+1}$ such that $S_0,S_1,\dots,S_{p-1}$ are pairwise disjointed subsets of $V_G$ and $(S_p,S_{p+1})$ partitions $S_0$, and that the maximum $p$ that will occur in the execution of Algorithm \ref{algo:main algorithm} is exactly $ctw(\mathbb{H})=ctww(H)$. Also, checking if $(S_1,\dots,S_{p-1},S_p,S_{p+1}) \preceq_{\R} (T_1,\dots,T_{p-1},T_p,T_{p+1})$ can be performed in $O(|V_G|^2)$ time.

\end{proof}

\section{Soundness of Algorithm \ref{algo:main algorithm parameterized}}
\label{sec:main algorithm parameterized}

\begin{algorithm}
\caption{Solving $\Omega$($H$-($\R$-MORPHISM)) ie. $\Omega$(BINARY-CSP($\Gamma$)) (parameterized version)}
\label{algo:main algorithm parameterized}
Create a tabular $\overline{\text{OMEGA}}$ filled with $0_A$

\For{$s\in V_G$, $t\in V_H$}
{
\If{$(l_G(s,s),l_H(t,t))\in\R$}
{$\overline{\text{OMEGA}}\begin{bmatrix} \{s\} & \{t\} \end{bmatrix}\gets \Omega_W(\{f^{\{s\}}_{\{t\}}\})$ (with $f^{\{s\}}_{\{a\}} : \begin{matrix} \{s\} & \mapsto & \{a\} \\ s &\mapsto & a \end{matrix}$)}
}
\For{$k=n-1$ downto $1$}
{
$(S_p,S_{p+1})\leftarrow$ contracted pair in the contraction $G_{k+1} \rightarrow G_k$ of $\mathbb{G}$
$S_0 \leftarrow$ contraction of $S_p$ and $S_{p+1}$ in $G_k$
$C=\{S_0,S_1,\ldots S_{p-1}\}$ the $\red$-connected component of $G_k$ containing $S_0$
$C_1 \uplus \ldots\uplus C_q = \{S_1,\ldots,S_{p-1},S_p,S_{p+1}\} $ be the partitionning of $(C\setminus \{S_0\})\cup \{S_p,S_{p+1}\}$ into $\red$-connected component in $G_{k+1}$

\For{$j=1$ to $q$}
{
Define $I_j\subseteq [p]$ such that $C_j=\{S_i, i\in I_j \}$

$I_j=:\{ I_j[1],\dots,I_j[p_j] \}$ with $p_j=|I_j|$
}
\For{$T_0,T_1,\ldots,T_{p-1}\subseteq V_H$ with $\emptyset \notin \{T_0,T_1,\ldots,T_{p-1}\} $}
{
\For{$T_p\cup T_{p+1}=T_0$ with $\emptyset\notin \{T_p,T_{p+1}\}$}
{
\If{$(T_1,\dots,T_{p-1},T_p,T_{p+1}) \succeq_{\R} (S_1,\dots,S_{p-1},S_p,S_{p+1})$}
{
$\overline{\text{OMEGA}}\begin{bmatrix} S_0 & T_0 \\  S_1 & T_1 \\ \vdots & \vdots  \\ S_{p-1} & T_{p-1}\end{bmatrix}+=$ $\overline{\text{OMEGA}}\begin{bmatrix} S_{I_1[1]} & T_{I_1[1]} \\  \vdots & \vdots  \\ S_{I_1[p_1]} & T_{I_1[p_1]}\end{bmatrix}\times\dots\times$ $\overline{\text{OMEGA}}\begin{bmatrix} S_{I_q[1]} & T_{I_q[1]} \\  \vdots & \vdots  \\ S_{I_q[p_q]} & T_{I_q[p_q]}\end{bmatrix}$
}
}
}
\Return{$\sum\limits_{T\subseteq V_H} \overline{\text{{\em OMEGA}}}\begin{bmatrix} V_G & T \end{bmatrix}$}
}
\end{algorithm}

\begin{definition}

Let $G$ and $H$ two $\red$-free edge-labelled graphs on $n$ and $m$ vertices respectively. Let $\R\subseteq X_G \times X_H$.
Let $\S=(S_1,\dots,S_p)$ be a tuple of $p$ pairwise disjointed subsets of $V_G$ (in particular, $S_1,\dots,S_p$ can be different vertices of a contraction $G'$ of $G$) and $S=S_1\uplus\dots\uplus S_p$.
Let $\T=(T_1,\dots,T_p)$ be a tuple of $p$ non-empty subsets of $V_G$ and $T=T_1\uplus\dots\uplus T_p$.

We then define the set $\overline{\R}^{\S}_{\T}$ by \[\overline{\R}^{\S}_{\T} = \{ f:G[S]\underset{\R}{\rightarrow} H[T] \mid f(S_1)= T_1,\dots,f(S_p)= T_p \}.\]

\end{definition}

\begin{lemma}\label{lemma:empty set,join lemma parameterized}

Let $G$ and $H$ two $\red$-free edge-labelled graphs, on $n$ and $m$ vertices respectively, $G'$ be a contraction of $G$, $\R\subseteq X_G\times X_H$, $C_1,\dots,C_q$ be $q$ different $\red$-connected components of $G'$ ($q\geq 1$), $\{S_1,\dots,S_p\} = C_1\uplus\dots\uplus C_q$ let $\S=(S_1,\dots,S_p)$, $\T=(T_1,\dots,T_p)$ be a tuple of $p$ non-empty subsets of $V_H$. Let for all $j\in q$, $I_j\subseteq [p]$ be such that $C_j=\{S_i,i\in I_j\}$. Then

\[\overline{\R}^{\S}_{\T} = \left\{
    \begin{array}{ll}
        \emptyset  & \text{ if } \T\cancel{\succeq}_{\R}\S
        \\
        \overline{\R}^{\S_{I_1}}_{\T_{I_1}}\Join \dots\Join \overline{\R}^{\S_{I_q}}_{\T_{I_q}} & \text{ if } \T\succeq_{\R}\S
    \end{array}\right.\]

\end{lemma}

\begin{proof}

First, assume that $\T\cancel{\succeq}_{\R}\S$:

Assume by contradiction that there exists $f\in \overline{\R}^{\S}_{\T}$, ie. $f$ is an $\R$-morphism from the edge-labeled graph $G[S_1\uplus\dots\uplus S_p]$ to the edge-labeled graph $H$ and $\forall i\in [p], f(S_i) = T_i$. Since $\T$ does not make $\S$ $\R$-feasible, there exists a pair $(i,i')\in [p]^2$ such that:

\begin{itemize}
    \item $l_{G'}(S_i,S_{i'})\neq \red$,
    \item there exists $a\in T_i$ and $b\in T_{i'}$ with $(l_{G'}(S_i,S_{i'}),l_H(a,b))\notin \R$.
\end{itemize}

By Definition \ref{def:Meaning Contraction}, $\forall (u,v)\in S_i\times S_{i'}, l_G(u,v)= l_{G'}(S_i,S_{i'})$. Notice that, by definition of $f\in\overline{\R}^{\S}_{\T}$, since $f(S_i)=T_i$ and $f(S_{i'})=T_{i'}$, there exists $(u,v)\in S_i\times S_{i'}$ such that $(f(u),f(v))=(a,b)$. Since $f$ is an $\R$-morphism, we have $(l_G(u,v),l_H( f(u),f(v) ) ) = (l_{G'}(S_i,S_{i'}) , l_H(a,b) ) \in \R$, which contradicts the definition of $(a,b)$. We have a contradiction, which proves that $\overline{\R}^{\S}_{\T} = \emptyset$.

Second, assume that $\T\succeq_{\R}\S$:

Take $f\in \overline{\R}^{\S}_{\T}$ and $j\in [q]$. Since restrictions of $\R$-morphisms are $\R$-morphisms, $f|_{\cup\S_{I_j}}$ is an $\R$-morphism. Moreover, since $f\in \overline{\R}^{\S}_{\T}$, for all $j\in I_j$, $f(S_j) = T_j$. We deduce that $f|_{\cup\S_{I_j}}\in \overline{\R}^{\S_{I_j}}_{\T_{I_j}}$, which proves that $\overline{\R}^{\S}_{\T} \subseteq \overline{\R}^{\S_{I_1}}_{\T_{I_1}}\Join \dots\Join \overline{\R}^{\S_{I_q}}_{\T_{I_q}}$.

We now prove the reverse. Let $(f_1,\dots,f_q)\in \overline{\R}^{\S_{I_1}}_{\T_{I_1}}\times \dots\times \overline{\R}^{\S_{I_q}}_{\T_{I_q}}$, and let $f=f_1\Join\dots\Join f_q$. We will prove that $f\in \overline{\R}^{\S}_{\T}$. Clearly, by definition of the $\overline{\R}^{\S_{I_j}}_{\T_{I_j}}$ for $j\in [q]$, we have, for all $i\in [p], f(S_i) = T_i$ (knowing that $(I_1,\dots,I_q)$ is a partition of $[p]$). There only remains to prove that $f$ is an $\R$-morphism.

Let $S=S_1\uplus\dots\uplus S_p$, and $(u,v)\in S^2$. We will prove that $(l_G(u,v),l_H(f(u),f(v)))\in \R$. Let $(i,i')\in [p]^2$ be such that $u\in S_i$ and $v\in S_{i'}$.

\begin{enumerate}

\item If there exists $j\in [q]$ such that $(i,i')\in (I_j)^2$ (i.e.\ if $S_i$ and $S_{i'}$ belong to the same $\red$-connected component $C_j$), then, $(u,v)\in (S_{I_j})^2$. It follows by definition of $f$ that $(f(u),f(v))=(f_j(u),f_j(v))$, and then $(l_G(u,v),l_H(f(u),f(v)))\in\R$ since $f_j$ is an $\R$-morphism.

\item Else, by definitions of $I_j$ and $C_j$ for $j\in [q]$, $S_i$ and $S_{i'}$ are not $\red$-connected in $H'$. We deduce that, in particular, $l_{G'}(S_i,S_{i'})\neq \red$. Using Definition \ref{def:Meaning Contraction}: $\forall (u,v)\in S_i\times S_{i'}, l_G(u,v)=l_{G'}(S_i,S_{i'})$. Then, $(f(u),f(v))=(f_i(u),f_{i'}(v))\in T_i\times T_{i'}$, thus using the hypothesis of $\T\succeq_{\R}\S$, we have $(l_G(u,v),l_H(f(u),f(v))) = (l_{G'}(S_i,S_{i'}),l_H(f(u),f(v)))\in \R$.

\end{enumerate}

We have proven that $(l_G(u,v),l_H(f(u),f(v)))\in\R$. Hence, $f$ is indeed an $\R$-morphism, which concludes the proof.

\end{proof}

\begin{lemma}\label{lemma:split parameterized}

Let $G$ and $H$ be two $\red$-free edge-labelled graphs, $\R\subseteq X_G\times X_H$, $T_0,T_1,\dots,T_{p-1}$ be $p$ non-empty subsets of $V_H$, $\T=(T_1,\dots,T_{p-1})$, $S_0,S_1,\dots,S_{p-1}$ be $p$ pairwise disjointed subsets of $V_G$, $\S=(S_1,\dots,S_{p-1})$, and $S_p\uplus S_{p+1}=S_0$ be a partition of $S_0$. Then

\[\overline{\R}^{\S,S_0}_{\T,T_0}
    =
\underset{\begin{matrix} T_p\cup T_{p+1} = T_0 \\ (\T,T_p,T_{p+1}) \succeq_{\R} (\S,S_p,S_{p+1}) \end{matrix}}{\uplus}
\overline{\R}^{\S,S_p,S_{p+1}}_{\T,T_p,T_{p+1}}
\]

\end{lemma}

\begin{proof}

Partitioning the set
$\overline{\R}^{\S,S_0}_{\T,T_0}$ in equivalence classes with respect to the equivalence relation $\sim$ defined by:

\hfill

\[\forall (f,f')\in
(\R^{\S,S_0}_{\T,T_0})^2,
f\sim f' \iff (f(S_p),f(S_{p+1})) = (f'(S_p),f'(S_{p+1}))\]

we obtain:

\[\overline{\R}^{\S,S_0}_{\T,T_0}
    =
\underset{T_p\cup T_{p+1} = T_0}{\uplus}
\overline{\R}^{\S,S_p,S_{p+1}}_{\T,T_p,T_{p+1}}
\]

and then, the conclusion follows from Lemma \ref{lemma:empty set,join lemma parameterized}.

\end{proof}

\begin{theorem}\label{thm:main cardinal theorem parameterized}

Let $\Omega$ be a semiring pre-morphism and $W$ a weight matrix. Let $G$ and $H$ be two $\red$-free edge-labelled graphs, $(G_n,\dots,G_1)$ a contraction sequence of $G$ (with $n:=|V_G|$), $k\in [n-1]$, $C_1,\dots,C_q$ be $q$ ($q\geq 1$) different $\red$-connected components of $G_{k+1}$ ($q\geq 1$), $\{S_1,\dots,S_{p-1},S_p,S_{p+1}\} = C_1\uplus\dots\uplus C_q$, $\S=(S_1,\dots,S_{p-1})$, and $T_0,T_1,\dots,T_{p-1}$ be $p$ non-empty subsets of $V_H$, and let $\T=(T_1,\dots,T_{p-1})$. Denote for all $j\in q$, $I_j\subseteq [p+1]$ such that $C_j=\{S_i,i\in I_j\}$. Assume that $G_k$ is obtained from $G_{k+1}$ by merging the two different vertices $S_p$ and $S_{p+1}$ to the vertex $S_0=S_p\uplus S_{p+1}$. Then, we have:

\[\Omega_W(\overline{\R}^{\S,S_0}_{\T,T_0})
    =
\underset{\begin{matrix} T_p\cup T_{p+1} = T_0 \\ (\T,T_p,T_{p+1}) \succeq_{\R} (\S,S_p,S_{p+1}) \end{matrix}}{\sum} \underset{j=1}{\overset{q}{\prod}}
(\Omega_W(\overline{\R}^{(\S,S_p,S_{p+1})_{I_j}}_{(\T,T_p,T_{p+1})_{I_j}})\]

(where $\sum$ and $\prod$ refers to the sum and product of the semiring).

\end{theorem}

\begin{proof}

Using Lemma \ref{lemma:split parameterized}:

$\Omega_W(\overline{\R}^{\S,S_0}_{\T,T_0})
    =
\Omega_W(\underset{\begin{matrix} T_p\cup T_{p+1} = T_0 \\ (\T,T_p,T_{p+1}) \succeq_{\R} (\S,S_p,S_{p+1}) \end{matrix}}{\uplus}
\overline{\R}^{\S,S_p,S_{p+1}}_{\T,T_p,T_{p+1}}\,\,\,\,\,)
$

\hfill

using the first axiom of pre-morphisms:

\hfill

$\Omega_W(\overline{\R}^{\S,S_0}_{\T,T_0})
    =
\underset{\begin{matrix} T_p\cup T_{p+1} = T_0 \\ (\T,T_p,T_{p+1}) \succeq_{\R} (\S,S_p,S_{p+1}) \end{matrix}}{\sum} \Omega_W(\overline{\R}^{\S,S_p,S_{p+1}}_{\T,T_p,T_{p+1}})
$

\hfill

and we can iteratively apply the second axiom of pre-morphisms with $\F:=\overline{\R}^{\S,S_p,S_{p+1}}_{\T,T_p,T_{p+1}}$ and $\forall j\in [q], \F_j:=\overline{\R}^{(\S,S_p,S_{p+1})_{I_j}}_{(\T,T_p,T_{p+1})_{I_j}}$, thanks to Lemma \ref{lemma:empty set,join lemma} (since $(\T,T_p,T_{p+1})  \succeq_{\R}  (\S,S_p,S_{p+1})$) in order to write:

\hfill

$\Omega_W(\overline{\R}^{\S,S_0}_{\T,T_0})
    =
\underset{\begin{matrix} T_p\cup T_{p+1} = T_0 \\ (\T,T_p,T_{p+1}) \succeq_{\R} (\S,S_p,S_{p+1}) \end{matrix}}{\sum} \underset{j=1}{\overset{q}{\prod}}
\Omega_W(\overline{\R}^{(\S,S_p,S_{p+1})_{I_j}}_{(\T,T_p,T_{p+1})_{I_j}})$,

\hfill

which was what we wanted to prove.

\end{proof}

\hfill

\begin{lemma}\label{lemma:loop invariant parameterized}

Let $(A,+,\times,0_A,1_A)$ be a semiring, $B$ a set, and $\Omega$ a $A$-pre-morphism with weights in $B$ (with $A$ and $B$ sets), $W\in B^{V_G\times V_H}$, $G$ and $H$ be two $\red$-free edge labelled graphs, $(G_n,\dots,G_1)$ be a contraction sequence of $H$.
    
Then, for all $k\in [m]$, for all $\red$-connected component $\{S_1,\dots,S_p\}$ of $G_k$, and for all $T_1,\dots,T_p$ non-empty subsets of $V_H$, after the iteration of index $k$ (and before the iteration of index $k-1$) of the loop ``for $k=m-1$ downto 1'' of Algorithm \ref{algo:main algorithm parameterized}, $\overline{\text{OMEGA}}\begin{bmatrix} S_1 & T_1 \\  \vdots & \vdots  \\ S_p & T_p\end{bmatrix}$ contains $\Omega_W(\overline{\R}^{S_1,\ldots,S_p}_{T_1,\ldots,T_p})$.

\end{lemma}

\begin{proof}

We proceed by induction over $k=m$ downto $1$.

\underline{Initialisation:} For $k=m$, it comes the loop: ``for $T\subseteq V_G$'', noticing that the $\red$-connected components of $G_n$ are exactly its singletons of vertices, ie, the $\{s\}$ for $s\in V_G$, and noticing that, for all $T\subseteq V_H$, $\overline{R}^{\{s\}}_{T} = \{f^{\{s\}}_{\{a\}}\}$ if $T$ is a singleton containing $T=:\{a\}$ such that $\forall (u,v)\in S^2, (l_G(s,s),l_H(a,a))\in \R$, and $\overline{\R}^{\{s\}}_{T}=\emptyset$ otherwise, recalling that $\Omega_W(\emptyset)=0_A$ (third axiom of pre-morphisms).

\underline{Hereditary:} Assume the lemma is true for $k+1$. Assume that the vertices merged in the contraction $G_{k+1}\rightarrow G_k$ are called $S_p$ and $S_{p+1}$ and that the merged vertex is called $S_0$. Notice that the only $\red$-connected component of $G_k$ that is not also a $\red$-connected component of $G_{k+1}$ is the one that contains $S_0$. Call it $C=\{S_0,S_1,\dots,S_{p-1}\}$. Let $C_1 \uplus \ldots\uplus C_q = \{S_1,\ldots,S_{p-1},S_p,S_{p+1}\} $ be the partitionning of $(C\setminus \{S_0\})\cup \{S_p,S_{p+1}\}$ into $\red$-connected component in $G_{k+1}$. For every $j\in [q]$, recalling that $C_j=\{ S_{I_j[1]},\dots,S_{I_j[p_j]}\}$ (by definition of $I_j$, and denoting $p_j:=|I_j|$), we have by induction hypothesis: $\overline{\text{OMEGA}}\begin{bmatrix} S_{I_j[1]} & T_{I_j[1]} \\  \vdots & \vdots  \\ S_{I_j[p_j]} & T_{I_j[p_j]}\end{bmatrix}$ contains $\Omega_W(\overline{\R}^{\S_{I_j}}_{\T_{I_j}})$, where, we recall, $\overline{\R}^{\S_{I_j}}_{\T_{I_j}}$ is a notation for $\overline{\R}^{S_{I_j[1]},\dots,S_{I_j[p_j]}}_{T_{I_j[1]},\dots,T_{I_j[p_j]}}$. Then, we see that, at the end of the loop of index $k$ (and thus at the beginning of the loop of index $k-1$): ``for all $T_0,\dots,T_{p-1}$ subsets of $V_H$'': $\overline{\text{OMEGA}}\begin{bmatrix} S_0 & T_0 \\  S_1 & T_1 \\ \vdots & \vdots  \\ S_{p-1} & T_{p-1}\end{bmatrix}$ contains $\underset{\begin{matrix} T_p\cup T_{p+1} = T_0 \\ (\T,T_p,T_{p+1})  \succeq_{\R}  (\S,S_p,S_{p+1}) \end{matrix}}{\sum}
    (\overset{q}{\underset{j=1}{\prod}} \Omega_W (\overline{\R}^{(\S,S_p,S_{p+1})_{I_j}}_{(\T,T_p,T_{p+1})_{I_j}}))$ which, by Theorem \ref{thm:main cardinal theorem parameterized}, equals $\Omega_W(\overline{\R}^{\S,S_0}_{\T,T_0})$. $\overline{\text{OMEGA}}\begin{bmatrix} S_0 & T_0 \\  S_1 & T_1 \\ \vdots & \vdots  \\ S_{p-1} & T_{p-1}\end{bmatrix}$ contains $\Omega_W(\overline{\R}^{\S,S_0}_{\T,T_0})$ at the end of the loop of index $k$. Recall that $\overline{\R}^{\S,S_0}_{\T,T_0}$ is a notation for $\overline{\R}^{S_1,\dots,S_{p-1},S_0}_{T_1,\dots,T_{p-1},T_0} = \overline{\R}^{S_0,S_1,\dots,S_{p-1}}_{T_0,T_1,\dots,T_{p-1}}$. This concludes the proof.

\end{proof}

\parametrized*

\begin{proof}

The complexity comes from the fact that there are $(2^{|V_H|}-1)^{p+1}$ ways to choose $(T_0,T_1\dots,T_{p-1})$ and $(T_p,T_{p+1})$ such that $(T_0,T_1\dots,T_{p-1})$ are non-empty subsets of $V_H$ and $(T_p,T_{p+1})$ are non-empty subsets of $V_H$ with $T_p\cup T_{p+1} = T_0$.
Note that, using the axiom relative to $\uplus$ and $+$ and the hypothesis that $\Omega$ is corestriction independent, it follows that for all $\F\in T^S$ with $\forall f\in \F, f(S)\subseteq T'$,  denoting $\F|^{T'} = \{f|^{T'} , f\in\F \}$, we have $\Omega_W(\F)=\Omega_W(\F|^{T'})$. 

The soundness of Algorithm \ref{algo:main algorithm parameterized} follows from the observation that, partitioning the set of solution SOL of the instance $G$ of $H$-($\R$-MORPHISM):

\[\text{SOL}=\underset{T\subseteq V_H}{\uplus} \{ f\in \R^{V_G}_{V_H} \mid f(V_G)=T \}\]

which leads to

\[\Omega_W(\text{SOL})= \sum\limits_{T\subseteq V_H} \Omega_W(\{ f\in \R^{V_G}_{V_H} \mid f(V_G)=T \})\]

Using the above remark:

\[\Omega_W(\text{SOL})= \sum\limits_{T\subseteq V_H} \Omega_W(\{ f\in \R^{V_G}_{V_H} \mid f(V_G)=T \}|^{T})\]

from where we deduce:

\[\Omega_W(\text{SOL})= \sum\limits_{T\subseteq V_H} \Omega_W(\overline{\R}^{V_G}_T)\]

To obtain the soundness of Algorithm \ref{algo:main algorithm parameterized}, it only remains to notice that, since $V_G$ is a connected component of $G_1$, we deduce from Lemma \ref{lemma:loop invariant parameterized} that $\overline{\text{OMEGA}}\begin{bmatrix} V_G & T \end{bmatrix}$ contains exactly $\Omega_W(\overline{\R}^{V_G}_T)$ for all $T\subseteq V_H$.

\end{proof}

\section{Examples of Pre-Morphisms and Associated Problems}
\label{app:examples of pre-morphisms}

In this section we introduce additional examples of pre-morphisms and explicitly show how to formulate various BINARY-CSP$(\Gamma)$ problems in the $\Omega$($H$-($\R$-MORPHISM)) framework.
One of the smallest example of a non-trivial pre-morphism is the following Boolean pre-morphism:

\begin{lemma}\label{lemma:decision dioid}

Let $\mathbf{2}=\{ 0 , 1 \}$ be the set of the two Booleans and $\vee$ and $\wedge$ denote disjunction and conjunction over $\mathbf{2}$. Let $\mathbf{1}_{\neq\emptyset}$ be the function that maps $\emptyset$ to $0$ and every other set $\F$ to $1$ (the weight matrix $W$ is inessential).
Then, $(\mathbf{2}, \vee, \wedge,0,1)$ is a semiring (even a dioid), and the function $\mathbf{1}_{\neq\emptyset}$ is a strong poly-time computable $\mathbf{2}$-pre-morphism ignoring weights.
\end{lemma}

Note that computing the value by the function $\mathbf{1}_{\neq\emptyset}$ of the set of solutions of a BINARY-CSP (or a MORPHISM) instance is equivalent to solving its decision version. Therefore, the $\mathbf{1}_{\neq\emptyset}$(BINARY-CSP($\Gamma$)) and $\mathbf{1}_{\neq\emptyset}$($H$-($\R$-MORPHISM)) problems are exactly the BINARY-CSP($\Gamma$) and $H$-($\R$-MORPHISM) problems.
We can do better and implement the list version:

\begin{lemma}\label{lemma:list dioid}

Let list be the function that associates to every $\F\in\P(T^S)$ (for all sets $S$ and $T$) and $W\in \mathbf{2}^{S\times T}$: $0$ if $\{ f\in\F\mid \forall u\in S, W(u,f(u))=1 \}$ is empty, and $1$ otherwise.
Then ${\text{list}}$ is a poly-time computable strong $\mathbf{2}$-pre-morphism with weights in $\mathbf{2}$.

\end{lemma}

It is easy to see that the list(BINARY-CSP($\Gamma$)) and list($H$-($\R$-MORPHISM)) problems are exactly the list-BINARY-CSP($\Gamma$) and list-$H$-($\R$-MORPHISM) problems.
We can also define a pre-morphism around the function computing the cardinal of the set involved:

\begin{lemma}\label{lemma:counting dioid}

Let $\#$ be the function that associate to every set $\F$ the cardinal $|\F|$ of $\F$ (the weight matrix $W$ is inessential).
Then $(\N,+,\times,0_{\N},1_{\N})$ is a semiring (it is even a dioid), and \# is a poly-time computable $\N$-pre-morphism ignoring weights.

\end{lemma}

It follows that $\#$(BINARY-CSP($\Gamma$))  and $\#$($H$-($\R$-MORPHISM))  are alternative formulations of the $\#$BINARY-CSP($\Gamma$) and $\#H$-($\R$-MORPHISM) counting problems. 

To go even further, we can even add a list to the counting version.

\begin{lemma}\label{lemma:list counting dioid}

Let $\#_{\text{list}}$ be the function that associates to every $\F\in\P(T^S)$ and $W\in \mathbf{2}^{S\times T}$ (for all sets $S$ and $T$ and $W\in\mathbf{2}^{S\times T}$) the cardinal of the set of the functions $f$ of $\F$ satisfying $\forall u\in S, W(u,f(u))=1$: ${\#_{\text{list}}}_W(\F) = |\{ f\in\F\mid \forall u\in S, W(u,f(u))=1 \}|$.
Then $\#_{\text{list}}$ is a poly-time computable strong $\N$-pre-morphism with weights in $\mathbf{2}$.

\end{lemma}

The next semiring and pre-morphism are more subtle, since they actually takes into account the weight matrix with weights in $\Rbar$, and additionaly make use of a Cartesian product. Define \[\oplus:\begin{matrix} (\Rbar\times\N)^2 & \mapsto & \Rbar\times\N \\ (m_1,c_1),(m_2,c_2) & \mapsto & (\min(m_1,m_2), \left\{
    \begin{array}{ll}
        c_1 \text{ if } m_1<m_2 \\
        c_2 \text{ if } m_2<m_1 \\
        c_1+c_2\text{ if } m_1=m_2
    \end{array}\right\}  ) \end{matrix}\]

and \[\otimes:\begin{matrix} (\Rbar\times\N)^2 & \mapsto & \Rbar\times\N \\ (m_1,c_1),(m_2,c_2) & \mapsto & (m_1+m_2, c_1\times c_2  ) \end{matrix}\]

We immediately obtain the following lemma.
\begin{lemma}

$((\Rbar\times\N),\oplus,\otimes,(+\infty,0_{\N}),(0_{\Rbar},1_{\N}))$ is a semiring (it is even a dioid).

\end{lemma}

\begin{definition}
We define the following functions.
\begin{itemize}
    \item Let $\text{MinCost}$ be the function that maps, for $G$ and $H$ two edge-labelled graphs, $S\subseteq V_G$, $T\subseteq V_H$, any $(\F,W)$ to $\min W_{\sum}(\F)$, for $\F\in\P(T^S)$ and $W\in \Rbar^{V_G\times V_H}$ denoting, for all $f\in T^S,\  W_{\sum}(f)=\sum\limits_{u\in S} W(u,f(u))$, and $W_{\sum}(\F)=\{ W_{\sum}(f) \mid f\in\F \}$: $\text{MinCost}$ outputs values in $\Rbar$).

\item Let $\text{ArgMinCost}$ be the function that maps, for $G$ and $H$ two edge-labelled graphs, $S\subseteq V_G$, $T\subseteq V_H$, any $(\F,W)$ to the set $\left\{
    \begin{array}{ll}
        \emptyset \text{ if } \min W_{\sum}(\F)=+\infty
        \\
        \{f\in\F\mid W_{\sum}(f)=\min W_{\sum}(\F) \}\text{ otherwise }
    \end{array}\right\}$ with $\F\in\P(T^S)$ and $W\in\Rbar^{V_G\times V_H}$
    
\item Last, let $\#\text{ArgMinCost}$ be the composition of \# and $\text{ArgMinCost}$ (i.e., the function that outputs the cardinal of the set described above): $\#\text{ArgMinCost}$ outputs values in $\N$.
\end{itemize}
\end{definition}

\begin{lemma}\label{lemma:MinCost dioid}

$(\text{MinCost},\#\text{ArgMinCost})$ is a poly-time computable strong $(\Rbar\times\N)$-pre-morphism with weights in $\Rbar$.

\end{lemma}

Notice that computing the value by the function $(MinCost,\#ArgMinCost)$ of the set of solutions of a BINARY-CSP or a MORPHISM instance means computing the minimal weight $\sum\limits_{u\in V_G} w(u,f(u))$ of the solutions $f$ satisfying $\forall (u,v)\in V_G\times V_H, w(u,v)=+\infty \implies f(u)\neq v$, and determining the number of such solutions of minimal weights, which answers a problem that subsumes together the counting, the list, and a weighted version of BINARY-CSP($\Gamma$) and $H$-($\R$-MORPHISM).
We also remark that instead of considering the weights in $\Rbar$, considering any totally ordered set $B$ with an increasing, binary, associative and commutative operation $b:B^2\mapsto B$ (instead of $+$) would have been possible.
We present a variant of this pre-morphism, with the goal of modeling another weighted version of BINARY-CSP($\Gamma$) and $H$-($\R$-MORPHISM).

\begin{definition}

Let $\text{MinWeight}$ be the function that maps, for $G$ and $H$ two edge-labelled graphs, $S\subseteq V_G$, $T\subseteq V_H$, any $(\F,W)$ to $\min W_{\max}(\F)$, for $\F\in\P(T^S)$ and $W=(w(u,v))_{u\in V_G,v\in V_H}\in \Rbar^{V_G\times V_H}$ denoting, for all $f\in T^S,\  W_{\max}(f)=\sum\limits_{v\in T} \max\limits_{u\in f^{-1}(v)} W(u,v)$, and $W_{max}(\F)=\{ W_{\max}(f) \mid f\in\F \}$: $\text{MinWeight}$ outputs values in $\Rbar$).

Let $\text{ArgMinWeight}$ be the function that maps, for $G$ and $H$ two edge-labelled graphs, $S\subseteq V_G$, $T\subseteq V_H$, any $(\F,W)$ to the set $\left\{
    \begin{array}{ll}
        \emptyset \text{ if } \min W_{\max}(\F)=+\infty
        \\
        \{f\in\F\mid W_{\max}(f)=\min W_{\max}(\F) \}\text{ otherwise }
    \end{array}\right\}$ with $\F\in\P(T^S)$ and $W\in\Rbar^{V_G\times V_H}$.
    
and let $\#\text{ArgMinWeight}$ be the composition of $\#$ and $\text{ArgMinWeight}$ (ie the function that outputs the cardinal of the set described above): $\#\text{ArgMinWeight}$ outputs values in $\N$.

\end{definition}

Similarly, computing the value by $\text{MinWeight}$ of the set of solutions of an instance of BINARY-CSP($\Gamma$) or $H$-($\R$-MORPHISM), solves another weighted version of the problems, which is the one defined by Escoffier et al.~\cite{escoffier2006weighted} (restricted to the less general case of $q$-COLORING) while computing both the minimal weight of a solution and the number of solutions with such a minimal weight.

\begin{lemma}\label{lemma:MinWeight dioid}

$((\Rbar_+\times\N),\oplus,\otimes,(+\infty,0_{\N}),(0_{\Rbar_+},1_{\N}))$ is a semiring, and $(\text{MinWeight},\#\text{ArgMinWeight})$ is a poly-time computable $(\Rbar_+\times\N)$-pre-morphism (restricting and corestricting $\oplus$ and $\otimes$ to $(\Rbar_+\times\N)^2$ and $(\Rbar_+\times\N)$). It is not a strong pre-morphism.

\end{lemma}

Finally, we describe a last dioid, that adds a constraint on the preimages of the vertices of the target graph.

\begin{definition}
    
    Let $\#restrictive-list$ be the function that maps, for $G$ and $H$ two edge-labelled graphs, $S\subseteq V_G$, $T\subseteq V_H$, any $(\F,W)$ to $|\{ f\in\F \mid \forall v\in T, w(v)\neq +\infty \implies |f^{-1}(\{v\})|=W_1(v), \forall (u,v)\in S\times T, f(u)=v\implies W_2(u,v)=1 \}|$, for $\F\in\P(T^S)$ and $W=(W_1(v),W_2(u,v))_{u\in V_G,v\in V_H}\in (\overline{\N}\times\mathbf{2})^{V_G\times V_H}$ with $W_1=(W_1(v))_{u\in V_G,v\in V_H}\in \overline{\N}^{V_G\times V_H}$ a matrix of constant columns and $W_2\in\mathbf{2}^{V_G\times V_H}$.

\end{definition}

\begin{lemma}\label{lemma:restrictive dioid}

The $\#restrictive-list$ is a poly-time computable $\N$-pre-morphism with weights in $\overline{\N}\times\mathbf{2}$.

\end{lemma}

Computing the value by $\#restrictive-list$ of the set of solutions of an instance of BINARY-CSP($\Gamma$) or $H$-($\R$-MORPHISM), solves a stronger version of the list generalisation, which is the one defined by Diaz et. al \cite{diaz2005restrictive} (restricted to the less general case of $q$-COLORING), while also counting the number of solutions. Notice that, by taking weights in $\Nbar\times\Rbar$ (instead of simply $\Nbar\times\mathbf{2}$), we could have, similarly as previously, encoded a weighted version.

\bibliographystyle{abbrv}
\bibliography{biblio}
\end{document}